\documentclass[11pt,onecolumn,final, journal,transmag]{IEEEtran}

\usepackage[utf8]{inputenc}

\usepackage{graphicx}
\usepackage{amsmath}
\usepackage[version=4]{mhchem}
\usepackage{siunitx}

\usepackage{longtable,tabularx}
\usepackage{enumitem}

\usepackage{amsthm}
\usepackage{amssymb}
\usepackage{color}
\usepackage{xcolor}
\usepackage{soul}
\usepackage{float}
\usepackage{centernot}
\newtheorem{theorem}{Theorem}

\newtheorem{definition}{Definition}
\newtheorem{assumption}{Assumption}
\newtheorem{proposition}{Proposition}

\newcounter{AlfaBeaitCounter}[section]

\renewcommand{\theenumi}{\alph{enumi}}

\usepackage{algorithm}
\usepackage[noend]{algpseudocode}
\usepackage{bm}

\newcommand{\arctantwo}{\mathop{}\mathopen{}\mathrm{arctan2}\mathop{}}
\algnewcommand{\LeftComment}[1]{\Statex \(\triangleright\) #1}

\newcommand{\hlhl}[1]{{#1}}

\DeclareRobustCommand{\ns}[1]{{#1}}
\DeclareRobustCommand{\ds}[1]{{#1}}

\newcommand{\OPEN}{\textit{OPEN}}
\newcommand{\CLOSED}{\textit{CLOSED}}
\newcommand{\parent}{\textit{parent}}
\newcommand{\pathh}{\textit{path}}
\newcommand{\successors}{\textit{Successors}}

\newcommand{\LO}{\textit{LO}}
\newcommand{\LOMap}{\textit{LOMap}}

\begin{document}




\begin{center}
{\LARGE Altitude-Loss Optimal Glides in Engine Failure Emergencies -- Accounting for Ground Obstacles and Wind}

\bigskip\bigskip
{\Large Daniel~Segal, Aharon~Bar-Gill and  Nahum~Shimkin\footnote{Corresponding Author, e-mail shimkin@ee.technion.ac.il}}\\[8pt]
{ Viterbi Faculty of Electrical and Computer Engineering\\
Technion -- Israel Institute of Technology} 
\bigskip\bigskip

{April 13, 2023}
\end{center}

\bigskip
\begin{center}
\noindent\fbox{%
    \parbox{\textwidth-2.0in}{%
       \textit{\large Daniel Segal, the lead author of this paper and a former master's student of the two other authors, tragically deceased before this work was completed. The current paper presents the last version written by Daniel, with minor style modification.} 
        \medskip

        \textit{\large This work was a followup on Daniel's M.Sc.\ thesis, which was published as an article in the Journal of Guidance, Control and Dynamics (2019), and won the 2018 best graduate student paper award by the Israeli Association for Automatic Control.}    
        \medskip
        
        \textit{\large This arXiv publication is dedicated to the memory of Daniel Segal, an outstanding engineer, an accomplished researcher, and a friend.}
    }%
}
\end{center}

\bigskip\bigskip 

\begin{center}
\textbf{\Large Abstract} \\[12pt]
\parbox{\textwidth-0.4in}{%
Engine \ns{failure} is a recurring emergency in General Aviation and fixed-wing UAVs, often requiring the pilot or remote operator to carry out carefully planned glides to safely reach a candidate landing strip. \ns{We} tackle the problem of minimizing the altitude loss of a thrustless aircraft flying towards a designated target position. 
\ns{Extending previous work on optimal glides without obstacles, we consider here trajectory planning of optimal gliding in the the presence of ground obstacles, while accounting for wind effects.}    
Under \ns{simplifying model} assumptions, in particular neglecting the effect of turns, we characterize the optimal solution as comprising straight glide segments between \ns{iteratively-determined} extreme points \hlhl{on} the obstacles.
\ns{Consequently, the optimal trajectory is included in an iteratively-defined \emph{reduced visibility graph}, and can be obtained by a standard graph search algorithm, such as A$^*$.}
We further quantify the effect of turns to verify a safe near-optimal glide trajectory. We apply our algorithm on a Cessna 172 model, in realistic scenarios, demonstrating both the altitude-loss optimal trajectory calculation, and \ns{determination of} airstrip reachability.
 }%
\end{center}


\bigskip\bigskip

\section*{Nomenclature}

\noindent\begin{tabular}{@{}lcl@{}}
    ALO &=& \hlhl{Altitude Loss Optimal}\\
    \ns{FTP} &=& \hlhl{Free Tangent Point}\\
	$C_L$, $C_{D0}$ &=& Lift and profile drag coefficients\\
	$D$, $L$ &=& Drag and lift force\\
		$f_g$, $f_0$&=& Glide slope function and sink rate function, respectively\\
	$J$  &=& Cost function\\
	$K$ &=& Induced drag coefficient\\
			$m$ &=& Aircraft's mass\\

\end{tabular}

\noindent\begin{tabular}{@{}lcl@{}}
	$n$ &=& Load factor\\
	$P_A$ &=& Aircraft current position\\
	$P_B$ &=& The candidate landing site location\\
	$q$, $q_0$ &=& Dynamic pressure, optimal dynamic pressure in still air\\
	$S$ &=& Aerodynamic surface\\
	$t,t_f$  &=& Time of flight, total final time of flight from engine cut-off\\
	$V$ &=& True air velocity\\
	$V_g$ &=& Aircaft velocity relative to the ground\\
	$V^*$  &=& The minimum glide slope velocity\\
	$V_{opt}$  &=& The height-loss optimal flight velocity\\
	$V_0$ &=& Optimal velocity in still air\\
	$V_i$ &=& Graph Vertex $i$\\
	$V_{stall}$, $V_{max}$  &=& Aircraft's stall and maximum velocity\\
	$V_s$  &=& Initial flight velocity\\
	$\mathbf{W}$  &=& Air-mass velocity vector in Ground frame\\
	$W_{X}$, $W_{Y}$ &=& North and east components of the air-mass velocity\\
	$W_\perp$, $W_\parallel$  &=& Crosswind and in-plane wind\\
	$[X,Y,Z]^T$&=& Location coordinate vector, North, East, Down, in the Ground frame\\
	$\gamma$ &=& Vertical path angle in the Air-mass frame\\
	$\psi$ &=& Flight heading relative to the air-mass\\
	$\psi_g$ &=& Flight heading in the Ground frame\\
	$\phi$ &=& The bank angle\\
	$\rho$ &=& Air-density\\
\end{tabular}\\

\section{Introduction}
Aircraft engine \ns{failure} is a recurring emergency in civil aviation, and especially in General Aviation (GA) as well \ns{as} in UAV operations, \ns{possibly leading to total loss of thrust}.
\ns{A report by the Australian Transport Safety Bureau on engine failures in light aeroplanes indicates a failure rate of piston engines between 1 to 6 per 10000 flight hours} \cite{ref:engine_failures_2016}.  
Faced with such a predicament, the pilot must contend with the task of guiding the thrustless aircraft to a \ns{safe} landing strip, a task that may be further hindered by the presence of significant wind, elevated ground obstacles, and limited visibility conditions. This challenging task may be aided by \hlhl{an algorithm} that can determine in real-time the reachability of candidate landing strips, plan in real-time an optimized gliding trajectory to a selected landing site, and guide the pilot in following this trajectory. 

\ns{We address} this situation by considering the problem of Altitude-Loss Optimal (ALO) trajectories, namely finding gliding trajectories that minimize the altitude loss between a starting position and a target point.
The premise is that the aircraft will be able to lose any excess altitude when it approaches the designated landing strip by a suitable maneuver, and therefore trajectory planning is focused on reaching the destination point at maximal altitude. 
\ns{Moreover, solving for the minimal altitude-loss trajectory also resolves the issue of the reachability of a candidate landing point.} 

\ns{In a previous paper} \cite{ref:Segal}, ALO gliding trajectories were 
obtained under possibly significant wind conditions, albeit in the absence of obstacles.
The present paper extends this work by addressing terrain-induced obstacles, \ns{such as those determined by a ground elevation map,} that must be circumvented during the aircraft's thrustless descent towards the designated target position. \ns{A major challenge is to obtain an algorithm that reliably computes the required trajectory in real time, namely in a matter of seconds, to address an ongoing emergency situation.}
\ns{Such algorithms are to be incorporated in flight-management architectures that enable the avionics to assist the pilot or remote operator during emergency situations}, as discussed in \cite{ref:TotalLossofThrust, ref:AircraftLossOfControl2017}. 


\subsection{Literature Review}

\ns{Aircraft trajectory planning} and the related field of robotic path planning are wide topics with extensive literate (e.g., \cite{ref:LaVallePllaningAlg}). Recent interest has been largely motivated by UAV applications: Recent surveys can be found in \cite{survey2010,survey2019,survey2020}.  
We will focus here on works that address trajectory planning for emergency landing. Moreover, we mostly address works that consider this problem in conjunction with obstacle avoidance.
Additional references for optimal gliding without obstacles can be found in \cite{ref:Segal}.

Grid-based methods divide the configuration space into cells of given size or resolution, and search for the optimal path over the graph that connects adjoining cells. 
The scheme proposed in \cite{ref:Adler} uniformly discretized the state space, employed flight primitives to connect grid points, and utilized the Dijkstra algorithm for optimal graph search. This allows for avoiding ground obstacles, but the method is computationally inefficient. Papers \cite{ref:HeuristicGeneticAlgorithmApproaches,ref:EvaluatingHardwarePlatforms} suggest and examine the use of genetic algorithms for searching over a dense grid. 

Sampling-based motion planning algorithms have proved effective in high dimensional spaces. In \cite{ref:RRT_Star_AR,ref:EmergencyLandingGuidance} the authors employ variants of the RRT* algorithm 
for emergency landing path planning. Such probabilistic optimization methods can only provide probabilistic guarantees on their convergence times. 

From a control-theoretic viewpoint, in \cite{ref:ReachabilityBasedLanding} the authors formulate the emergency landing problem as a Hamiltonian-Jacobi-Bellman (HJB) reachability problem in 6DOF space. This problem can be generally solved numerically using gridding of the state space; however, the resulting high dimensionality leads to prohibitive computation times. A sub-optimal solution is proposed for certain sub-problems using the concept of flight primitives, partially relying on the formulation in \cite{ref:Adler}.  In \cite{ref:Fixed-WingUAV} the authors time-discretize the dynamic state equations, either in nonlinear or linearized form, and apply a general optimal control solver to compute the solution. The effectiveness of this scheme is demonstrated for short-range landing scenarios with no obstacles.

Roadmap methods for obstacle avoidance employ geometric constructs to designate specific points in the configuration space, and then restrict the search to the graph that connects these points.  The papers \cite{ref:MeuleauEmergencyLandingPlanner2009a,ref:MeuleauEmergencyLandingPlanner2009b,ref:MeuleauEmergencyLandingPlannerExpriment} consider path planning for emergency landing using visibility graphs. Starting with a 2D visibility graph, the authors propose heuristic extensions of that graph to the 3D space problem. Searching these graphs generally leads to suboptimal trajectories. This general approach is akin to ours, however the emphasis in the present paper is on characterizing and finding \emph{optimal} paths, under our modeling assumptions and more specific objective function. Related work in \cite{Ayhan2019} suggests preflight contingency planning for engine failure: Trajectories that avoid no-fly zones are determined by setting way-points, which are connected by wind-dependent trochoidal paths.  We note that the visibility graph approach has also been used for 3D path planning for \emph{powered} aircraft, see for example the recent papers \cite{Frontera17, Ahmad17, Majeed18, Damato19, Huang19} and references therein.

\subsection{Main Contribution and Paper Outline}

In this work, we establish analytic results that allow deriving an efficient algorithm that computes the ALO gliding trajectory, subject to constant winds, and in the presence of general \hlhl{shaped terrain-induced obstacles.
Similar} to \cite{ref:Segal}, we employ an approximate, \hlhl{problem-specific aerodynamic model of the aircraft; this model uses the aircraft} speed (or angle of attack) and roll angle as instantaneous controls. As shown in \cite{ref:Segal}, the 
ALO solution between two points in the absence of obstacles is a fixed-heading trajectory with constant velocity (the magnitude of which depends on the wind intensity and direction). This free-space optimal fixed-heading glide solution serves as the basic component of the optimal trajectory, which generally comprises straight flight segments between obstacles. Our analysis initially neglects the effect of turns, assuming that the required direction changes can be carried out instantaneously and with no altitude loss. 

The proposed algorithm utilizes the roadmap approach for obstacle avoidance, by adapting and extending the visibility graph approach (e.g., \cite{ref:LaVallePllaningAlg}) which is known to be optimal for planar problems.
This approach allows to search only through salient points on the obstacles, rather than creating and searching through a dense discretization of the entire space.
\hlhl{The underlying idea is to create a sparse search graph, whose intermediate vertices constitute on-contour points on the ground obstacles; the links between these points correspond to fixed-heading glide segments, which comprise the trajectories to candidate landing sites.} This graph is generated iteratively, starting from an initial point of known altitude, \hlhl{the engine cutoff point}. 

Our results imply that a shortest path search \hlhl{over} the \hlhl{generated} graph indeed yields the required ALO trajectory. Optimality does not depend on the shape of the obstacles; hence the method is applicable for general terrain maps. 

To further account for \hlhl{the effect} of aircraft turns between straight flight segments, we develop an estimate for the altitude loss associated with such turns. These estimates can either be used to provide safer elevation loss guarantees for a previously-computed trajectory, or more \hlhl{generally -- be superposed onto} the search graph nodes (as these correspond to aircraft turns) during the graph building and search process. \hlhl{Thus, near-optimal trajectories, that also accommodate the effect of turns, are obtained.} 

\hlhl{To summarize, the main contributions of this paper are: (1) A method to calculate a local 2D obstacles map via intersections between wind-induced manifolds and the terrain ahead and then finding the on-contour extreme points to serve as grid nodes; solving OGS algorithmics over such evolving grids enables to optimally bypass obstacles of general shape. (2) We obtain a novel result, Theorem 3 in Subsection III.B, which states that it is enough to compute only two points to circumvent any terrain-induced local obstacle. (3) We convert the 3D obstacles problem into a sequence of 2D problems over a digital map and solve them via employing the A$^*$ algorithm. We show that the trajectory obtained from this algorithmic procedure is indeed the optimal trajectory in terms of altitude loss, subject to wind and terrain elevations. (4)  We formulate an extension to the proposed algorithm to consider the effect of turns.}


The paper is organized as follows. \hlhl{In Section} \ref{sec:ProblemFormulation} we formulate our optimization problem, including the notion of terrain-induced obstacles and wind modeling. 
In Section \ref{sec:SolutionConcept} we establish some theoretical results that \hlhl{substantiate} our approach, and present the proposed ALO trajectory planning algorithm that accounts for ground obstacles and winds. \hlhl{To generate the terrain-induced obstacles, we employ uniformly-spaced digital terrain mapping.}
In Section \ref{sec:EffectOfTruns} we extend our approach to consider the effect of turns and obtain sub-optimal trajectories. Next, in Section \ref{sec:DemoScenarios} we demonstrate our results on realistic engine cutoff scenarios. We summarize our main contributions and conclusions in section \ref{sec:Conclusion}.
Appendices \ref{subsec:AerodynamicModel}, \ref{subsec:AltitudeLossDuetoTurns} and \ref{subsec:NumericalExperiments} include some additional details on modeling and \hlhl{sample performance analyses.}
Appendix \ref{subsec:Flight} describes an initial flight experiment that was conducted with a Cessna 172 aircraft, towards validating the proposed concept and algorithm.   

\section{Model and Problem Statement} \label{sec:ProblemFormulation}
We proceed to present our trajectory planning problem and its mathematical modeling, outline our assumptions, and formulate the optimization problem.
Our interest is in computing an efficient gliding trajectory, in terms of minimal altitude loss, from a current location $\mathbf{A}$ to a candidate landing-strip $\mathbf{B}$, subject to wind and terrain-induced obstacles. The aerodynamic model that underlies this work is similar to that of \cite{ref:Segal}, and is briefly outlined in Appendix \ref{subsec:AerodynamicModel} for completeness.
\subsection{Frames of Reference}
We denote the Ground frame as a Local Level Local North (LLLN) inertial frame, with $X$ pointing to the North, $Y$ to East, and $Z$ downwards, i.e., NED coordinates. Its origin is located at the projection of the aircraft center of mass on the ground at sea level altitude at time $t=0$. Throughout this work, we consider a constant wind vector $(W_{X},W_{Y})$ of the air-mass relative to the ground frame. 
Noting our approximation of $\gamma\approx 0$ (see Appendix \ref{subsec:AerodynamicModel}, Equations (\ref{eq:X_dot},\ref{eq:Y_dot})), we have 
\begin{align*}
\dot{X} &= V\cos(\psi) + W_{X}\\
\dot{Y} &= V\sin(\psi) + W_{Y}
\end{align*}
when $\psi$ and $V$ are the aircraft heading and velocity relative to the air-mass and $\dot{X},\dot{Y}$ are the aircraft velocity north and east components in the Ground Frame.

The flight heading $\psi_g$ in the Ground frame is given by
\begin{equation}\label{eq:HeadingSlopeEq}
\psi_g = \arctantwo(\dot{Y},\dot{X})
\end{equation}
where $\arctantwo$ is the standard four-quadrant arctangent. 
Note that maintaining constant velocity $V$ and constant flight direction $\psi$ implies constant ground heading, $\psi_g$, in the Ground frame. 

\subsection{Terrain-Induced Obstacles}
Terrain-induced obstacles, or ground obstacles, can be naturally represented by an elevation map. We employ a digital map, e.g., the Shuttle Radar Topography Mission (SRTM) database \cite{ref:USGS} as the source for the elevation data. SRTM provides elevation values with a spatial resolution of about 30 m and elevation accuracy better than 16 m. By performing standard interpolation between the nearest samples in the discrete map $\textit{dtm}[m,n]$, we can produce a continuous elevation function, $\textit{dtm}(X,Y)$. The elevation map imposes the constraint
\begin{equation*}
-Z>\textit{dtm}(X,Y)+\textit{Clearance}
\end{equation*}
on the feasible trajectory. The \textit{Clearance} is meant to provide a safe distance from the ground. It should include the altitude error of the elevation
data, the aircraft instruments altitude error and the trajectory tracking error. We use the notation $\textit{DTM}(X,Y)$ to denote the elevation map augmented by this clearance; that is $\textit{DTM}(X,Y) \triangleq \textit{dtm}(X,Y)+\textit{Clearance}$.

\subsection{The Optimization Problem}\label{subsec:OptimizationProblem}
In the basic problem considered in this paper, we are given the current aircraft position $P_A$ and elevation $Z_A=Z(0)$, and a candidate landing site location, $P_B$. For attainability, we require to minimize the altitude loss between the current location $P_A$ and the destination location $P_{B}$.

\hlhl{We employ the model, derived in} \cite{ref:Segal} \hlhl{and detailed also in Appendix} \ref{subsec:AerodynamicModel}. \hlhl{This derivation is subject to:}
\begin{assumption}\label{asm:NoTurns}\
(a) We remove the constraint on $\dot{\psi}$ (Equation (\ref{eq:QuasiXiDotTan})), allowing $\psi$ to change freely. (b) The effect of boundary conditions in terms of initial and final velocity vectors is neglected. Consequently, both the aircraft velocity and pose at the initial and final points are not constrained. (c) The velocity control variable is limited to the feasible flight envelope of the aircraft ($V_{stall}(n) \leq V \leq V_{max}$). (d) We consider constant air-density. (e) We apply an optimistic cost on turns ($\phi=0$). (f) We consider the long-duration segment of the two time-scales of the problem; therefore, the change rate of the fast variables are eliminated: $\dot{\gamma} \cong 0$ and $\dot{V} \cong 0$. (g) We adopt the small-angle assumption for the FPA,  $\cos(\gamma) \approx 1$.
\end{assumption} 
\hlhl{Bringing in from Appendix} \ref{subsec:AerodynamicModel} \hlhl{the derivation result for the sink-rate function:} 
\begin{equation*}
\dot{Z} = f_0(V,\phi)=K_{SR}\left( \frac{V^4+n(\phi)^2V_0^4}{V} \right)
\end{equation*}
where $K_{SR}= \frac{\rho S C_{D0}}{2mg}$, $V_0=\sqrt{\frac{2mg}{\rho S}\sqrt{\frac{K}{C_{D0}}}}$, and $n(\phi) = \frac{1}{\cos(\phi)}$. 
The stall limit Eq.\ (\ref{eq:VstallEq}):
\begin{equation*}
V_{stall}(n(\phi))=\sqrt{\frac{2mg}{\rho S C_{Lmax}}n(\phi)}
\end{equation*}

The control variables in our model are the flight velocity, $V(t)$, and the flight heading relative to the air-mass, $\psi(t)$, while the variables, $X$, $Y$ and $Z$ are the state variables
in our problem. We define the cost as the altitude loss from $P_A$ to $P_B$. As we neglect the effect of turns on altitude-loss, we nullify the bank-angle variable in the sink-rate function, Eq.\  (\ref{eq:ImprovedSinkrateEquation}), namely, $\dot{Z}=f_0(V(t),\phi(t)) \cong f_0(V(t),0)$. Therefore, the altitude loss is given by the integral of sink rate $\dot{Z}=f_0(V(t),0)$ along the trajectory:
\begin{equation}\label{eq:CostFunctionEquation}
J(V)=Z(t_f)-Z(0)=\int_{0}^{t_f}f_0(V(t),0)dt
\end{equation}
We can now state our optimization problem:
\begin{equation}\label{eq:OptimizationProblem}
\begin{array}{cl}
\displaystyle \min_{V(t),\psi(t)}   & J(V) \\
\text{subject to} &  \dot{X} = V \cos(\psi) +W_{X}  \\
			& \dot{Y}  = V \sin(\psi) +W_{Y}\\
			& \dot{Z}=f_0(V(t),0)\\
			& V_{stall}(1) \leq V \leq V_{max}\\
			& -Z>\textit{DTM}(X,Y)\\
			& (X(0),Y(0)) =  P_A \\
			& Z(0)=Z_A\\
			& (X(t_f),Y(t_f))=P_{B}\
			\end{array}
\end{equation}
\hlhl{Note that Assumption} \ref{asm:NoTurns} \hlhl{
turns the model} Eq.\ (\ref{eq:X_dot})-Eq.\ (\ref{eq:xiDotEq}) \hlhl{into the approximate model in} Eq.\ (\ref{eq:OptimizationProblem}) here-above. \hlhl{Following our assumption that turns are instantaneous and do not incur altitude loss, our initial conditions do not include the aircraft orientation. We address this complementary effect in Section} \ref{sec:EffectOfTruns}.

Solving this optimization problem means that \hlhl{we aim at reaching the candidate landing site} $P_{B}$ \hlhl{with minimal} altitude loss.

\subsection{The ALO Free-Space Glide}\label{subsec:FreeSpaceGlide}
Let us first recall that the optimal gliding trajectory of an aircraft flying between two given points in fixed wind and in absence of obstacles, is indeed a straight path with fixed heading and speed. 
\begin{theorem}[Theorems 1 and 2 in \cite{ref:Segal}]\label{thrm:Theorem1}
	Consider the problem (\ref{eq:OptimizationProblem}) of an aircraft flying inside a constant velocity air-mass $(W_{X},W_{Y})$ subject to an imposed ground destination, as well as minimum (stall) and maximum velocity constraints. In the absence of ground obstacles, the optimal trajectory in the sense of minimal altitude loss must maintain a constant velocity and fixed heading. Furthermore, the optimal flight speed is given in Eq.\ (\ref{eq:Vopt}) below.
\end{theorem}
These straight path segments will serve as our building blocks for the optimal path in the presence of obstacles.  
The horizontal kinematics of an aircraft flying inside an air-mass with constant velocity and fixed-heading is illustrated in Fig.\ \ref{fig:FreeSapceFlight}. The aircraft must follow the ground track from the current aircraft \hlhl{location} $\mathbf{A}$ in the direction of the ground velocity vector, $\mathbf{V_g}$, subject to the wind vector $\mathbf{W}$  to reach its destination at \hlhl{location} $\mathbf{B}$.
\begin{figure}[H]
	\centering
	\includegraphics[width=6.2cm]{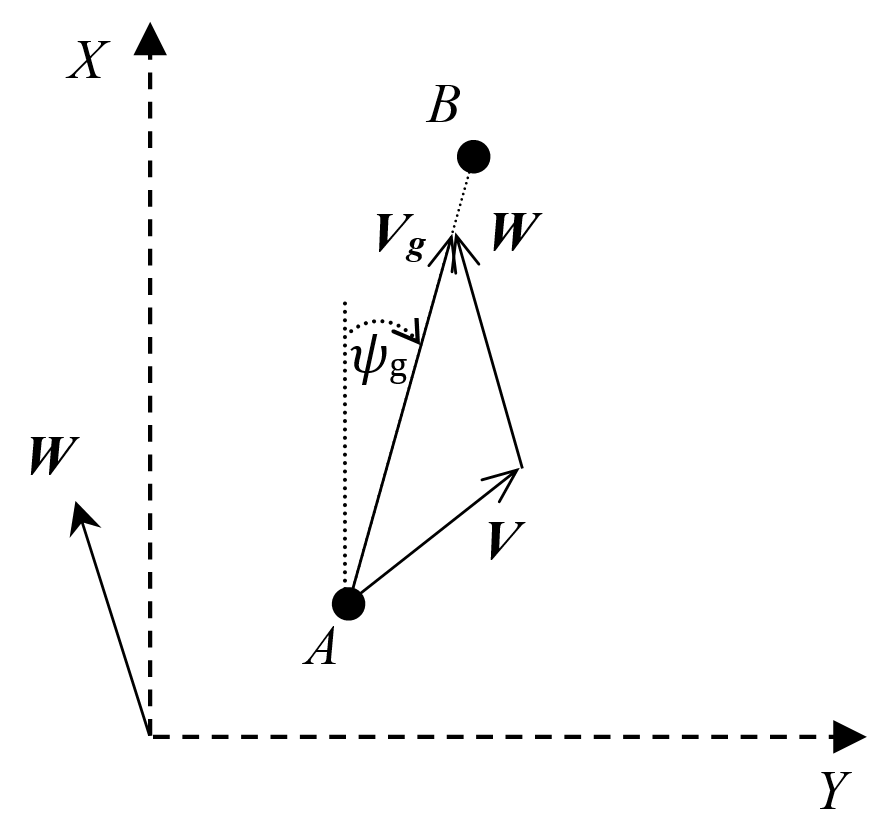}
	\caption{Fixed heading and fixed velocity glide analysis in free space -- top view}
	\label{fig:FreeSapceFlight}
	\par
\end{figure}
In constant-speed and fixed-heading flight, the aircraft speed in the Ground frame, $V_g$, is given by:
\begin{equation}
V_g=||\mathbf{V}+\mathbf{W}||
\end{equation}
In the absence of obstacles, the ALO flight according to Theorem \ref{thrm:Theorem1} is given by the fixed-heading, $\psi_g$, from $\mathbf{A}$ to $\mathbf{B}$. 
 
The wind velocity, $\mathbf{W}$, can be expressed in terms of the "in-plane" component, $W_\parallel$, in the direction of the flight heading in the Ground frame, $\psi_g$, and the "crosswind"  component, $W_\perp$, perpendicular to the flight heading, namely in direction $\psi_g+\frac{\pi}{2}$:
\begin{align}\label{eq:crossWindEquality}
W_\perp = W_\perp(\psi_g)     &=-W_{X}\sin(\psi_g)+W_{Y}\cos(\psi_g)\\
W_\parallel = W_\parallel(\psi_g) &=W_{X}\cos(\psi_g)+W_{Y}\sin(\psi_g)\label{eq:TailWindEquality}
\end{align}
The positive in-plane wind is equivalent to "tailwind", and the negative in-plane wind is equivalent to "headwind". Employing Eqs.\ (\ref{eq:crossWindEquality})-(\ref{eq:TailWindEquality}), the ground velocity, $V_g$, can be expressed as
\begin{equation}\label{eq:ClosingVelocty}
V_g=\sqrt{V^2-W_\perp ^2}+W_\parallel
\end{equation}
As shown in Appendix \ref{subsec:AerodynamicModel}, Eqs.\ (\ref{eq:BasicSinkRateEquation}) and (\ref{eq:ImprovedSinkrateEquation}), the altitude loss rate under these conditions is given by
$ \dot{Z} = f_0(V,0) =V_g f_g(V)$
where $f_g(V)$ is the glide slope function. This function is given explicitly in terms of the glide velocity as
\begin{equation}\label{eq:ExplicitFvGlideSope}
f_g(V) = \frac{K_{SR} \frac{V^4+V_0^4}{V}}{ \sqrt{V^2-W_\perp^2}+W_\parallel }
\end{equation}
where $K_{SR}= \frac{\rho S C_{D0}}{2mg}$ and  $V_0$ denotes the optimal velocity in still air (Equation (\ref{eq:OptimalVelocityinStillAir})).

To obtain the minimum of $f_g(V)$, i.e., the optimal (ALO) glide slope, we numerically solve the following sixth-degree, speed-to-fly equation \cite[Eq.\ (36)]{ref:Segal} 
\begin{equation}\label{eq:SpeedToFly}
V^6- \frac{3}{2} V^4 W_\perp^2+\frac{1}{2} W_\parallel\sqrt{V^2-W_\perp^2} \left( 3 V^4 -V_0^4 \right) -V^2 V_0^4+\frac{1}{2}  W_\perp^2 V_0^4=0
\end{equation}
which has a unique solution for $V \in (V_b,\infty)$, where $V_b = \sqrt{W_\perp^2 + \max(0,-W_\parallel)^2}$ is the minimal speed for which $V_g$ in Eq.\ (\ref{eq:ClosingVelocty}) is positive and well defined.

To satisfy the minimum and maximum velocity constraints, $(V_{stall},V_{max})$, we must limit the ALO glide velocity. Thus, the ALO glide velocity, $V_{opt}$, is given by:
\begin{equation}\label{eq:Vopt}
V_{opt}= V_{opt}(W_\parallel,W_\perp) =\min(\max(V_{stall},V^*),V_{max})
\end{equation}
while $V^*$ is the solution of the speed-to-fly equation (\ref{eq:SpeedToFly}). The ALO glide slope in free space flight in the direction $\psi_g$ is given by $f_g(V_{opt}(W_\parallel(\psi_g),W_\perp(\psi_g)))$.
\section{Solution Concept and Algorithm}\label{sec:SolutionConcept}
In this section, we derive the proposed algorithm to obtain the optimal trajectory in terms of minimal altitude-loss that avoids ground obstacles of general shape, in the presence of possibly intense wind. 
Starting from the initial position and altitude, our algorithm first creates a local 2D obstacle map by calculating the ALO straight-path trajectory to every map coordinate, and obtains the obstacles as those coordinates for which this trajectory is below ground level. Next, we find the set of {\it free tangent points} of these obstacles, which serve as the next vertices to be explored. In fact, we will show that at most two of these points need to be considered for each connected obstacle. We will further show that a fixed heading segment from the current position to one of these vertices must be included in the ALO trajectory. We may now iteratively continue to create a graph composed of such fixed heading segments, until the algorithm explores the destination and finds an ALO trajectory using a graph search algorithm.   

The proposed algorithm bears similarity to the classical 2D shortest path navigation problem which has been extensively explored in the literature, in particular for polygonal obstacles \cite{ref:LaVallePllaningAlg}.
However, in our problem, the so-called visibility road map cannot be calculated directly as the obstacles effectively depend on the current altitude of the aircraft. Therefore, the obstacle map depends on the currently explored vertex. Also, here we consider a non-Euclidean cost function, the altitude loss, which depends on the direction of the flight relative to the wind vector. We therefore justify our construction directly, based on Theorem \ref{thrm:Theorem1} above.

In Subsection \ref{subsec:ObtaclesAboeDTM} we describe the concept of the ALO manifold and define the obstacles and the free space. In Subsection \ref{subsec:TrajOptAlgorithm} we obtain the theoretical results that together with the definition of a feasible ("safe") path above a discrete digital map, Subsection \ref{subsec:ToDiscrete}, enable us to obtain the ALO trajectory optimization algorithm in Subsection \ref{subsecc:TheAloAlg}.

\subsection{The ALO Manifold and Local Obstacle Map}\label{subsec:ObtaclesAboeDTM}
In Subsection \ref{subsec:FreeSpaceGlide} we have identified the ALO glide velocity, $V_{opt}$, which leads to minimal altitude-loss for flight in a given heading $\psi_g$. Calculating the ALO glide slope $f_g(V_{opt})$ enables to obtain the minimal altitude-loss rate in every direction.

To obtain the relevant 2D obstacle map from the current position, it is convenient to first define the {\it ALO manifold}, $\textit{M}=\{\textit{M}(x,y)\}$, which is the cone-like surface of minimal altitude-loss to every displacement $(x,y)$.
Thus, $\textit{M}(x,y)$ is the altitude loss obtained by ALO glide from current \hlhl{position projection in the 2D horizontal plane,} $P$, to some \hlhl{displacement,} $P+(x,y)$, given by
 \begin{align}\label{eq:TheALOManfiold}
  \textit{M}(x,y)& = {||(x,y)||}\cdot{f_g(V_{opt}(W_\parallel(\psi_g),W_\perp(\psi_g)))}\\
 \psi_g &= \arctantwo(y,x)
 \end{align}
Here $f_g(V)$ and $V_{opt}$ are given by Eqs.\ (\ref{eq:ExplicitFvGlideSope}) and (\ref{eq:Vopt}). An example of the ALO manifold for a Cesna 172 model is given in Fig.\ \ref{fig:LandingStripAttainabilityThirdScnarioPaper1}.
\begin{figure}[H]
	\centering
	\includegraphics[width=7.2cm]{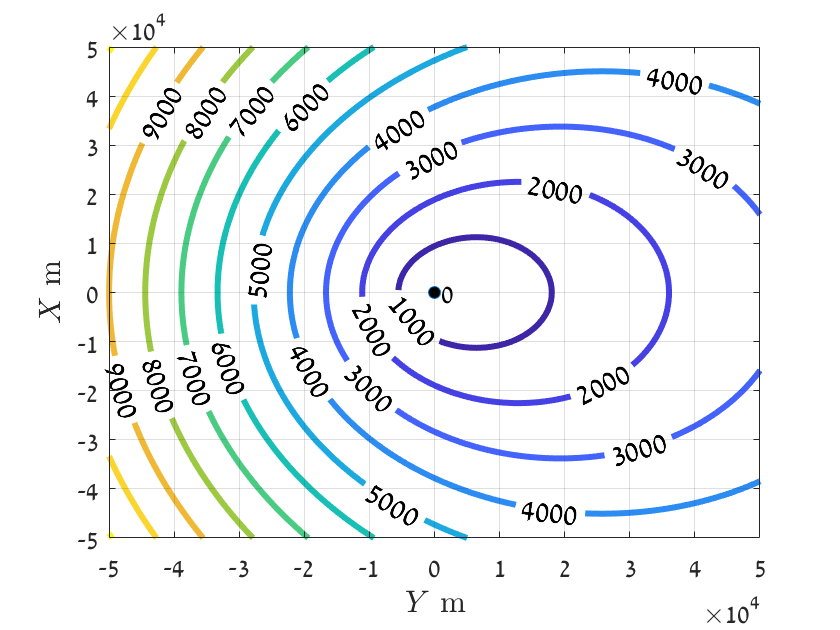}
	\caption{Altitude loss contours of the ALO manifold for a Cesna 172, given a 20 m/sec wind heading east}\label{fig:LandingStripAttainabilityThirdScnarioPaper1}
	\par
\end{figure}
We proceed to define the local obstacle function $\textit{LO}$.
For a given \hlhl{position projection onto the 2D horizontal plane,} $P_A =(X_A,Y_A)$, and altitude $Z_A$, we define the local obstacles function, $\textit{LO}(X,Y;P_A,Z_A)$, as
\begin{equation}\label{eq:LO}
\textit{LO}(X,Y;P_A,Z_A) = -Z_A-\textit{M}(X-X_A,Y-Y_A)-\textit{DTM}(X,Y)
\end{equation}
where the ALO manifold, $\textit{M}(x,y)$, is given by Eq.\ (\ref{eq:TheALOManfiold}). Thus, $\textit{LO}(X,Y;P_A,Z_A)$ is the elevation above ground level due to an ALO straight glide from $P_A$ to $(X,Y)$. An example of the obstacle function is given in Fig.\ \ref{fig:FullReachabilityW20Heading45Alt2000Again}, where we show an intersection between a DTM and an ALO manifold. In this example the aircraft is located at $P_A=(0,0)$ and altitude 2000 m above sea level. The regions where the ground elevation is above the ALO manifold are the local obstacles as viewed from the current position.
\begin{figure}[H]
	\centering
	\includegraphics[width=9.2cm]{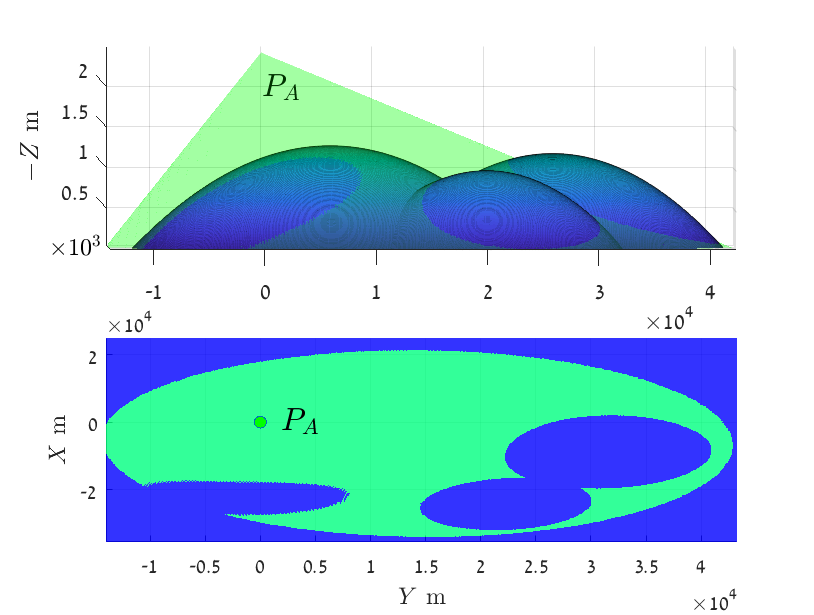}
	\caption{Local obstacle function, $\textit{LO}(X,Y;P_A,Z_A)$,  example}\label{fig:FullReachabilityW20Heading45Alt2000Again}
	\par
\end{figure}

Employing the definition of $\textit{LO}(X,Y;P_A,Z_A)$ we can define the 2D local obstacle map.
\begin{definition}[Obstacles and Free space]\
	
	Given the current position and altitude $(P_A,Z_A)$, define
	\begin{enumerate}[label=(\alph*)]
		\item Free Space $: \text{FREE} = \{(X,Y):\textit{LO}(X,Y;P_A,Z_A) \geq 0\}$
		\item Local Obstacles Set $: \text{OBST}=\{(X,Y):\textit{LO}(X,Y;P_A,Z_A) < 0\}$
		\item A single obstacle, $\mathcal{O}$, is a connected subset of the Local Obstacles Set.
	\end{enumerate} 
\end{definition}

\subsection{Trajectory Optimization -- Optimal Obstacle Avoidance}\label{subsec:TrajOptAlgorithm}
In this section, we lay out the theoretical foundation for the proposed trajectory optimization algorithm. Throughout this section we consider a fixed starting point $(P_A,Z_A)$, which stands for the current position and elevation of the aircraft along the planned trajectory.
In the following we refer to the \hlhl{state projection in the 2D horizontal plane,} $P=(x,y)$, as a \emph{point} or \emph{position}.
\begin{definition}[Convex Combination relative to $P_A$]\label{def:CCP}
	For a given point $P \in R^2$, the convex combination of $P$ to $P_A$ with parameter $\lambda \ge 0$ is $P(\lambda) \triangleq ((1-\lambda) P_A+ \lambda P)$; thus, $P(1)=P$.
\end{definition}
\begin{definition}[Direct Reachability]\label{def:DirectReach}
	The point $P$ is directly reachable from $P_A$ if we have $P(\lambda) \in \text{FREE}$ for all $\lambda \in [0,1]$.
\end{definition}

\begin{definition}[Obstacle Boundary Tangents]\label{def:OBTangents}
	\ds{For a point $P \in R^2$ and an obstacle, $\mathcal{O}$, when $P \in \partial \mathcal{O}$. Let $(x(s),y(s))$ be the parametrization of the curve $\partial \mathcal{O}$ around $P=(x(0),y(0))$.
		Let $(x'_+(0),y'_+(0))$ and $(x'_-(0),y'_-(0))$ be the derivative of the curve $(x(s),y(s))$ w.r.t. $s$ from $s=0^+$ and $s=0^-$ respectfully. Denote the obstacle boundary tangents of $P$ as $P^+(\lambda)$ and $P^-(\lambda)$ when  $P^+(\lambda)=P+\lambda(x'_+(0),y'_+(0)) $ and  $P^-(\lambda)=P+\lambda(x'_-(0),y'_-(0))$.}
\end{definition}

\begin{definition}[FTP -- Free Tangent Point]
\label{def:FTP}
	For a point $P_T \in R^2$, let $P_T(\lambda)$ denote its convex combination to $P_A$ as in Definition \ref{def:CCP}. A point $P_T$ is a free tangent point to an obstacle if 
	\begin{enumerate}[label=(\alph*)]
		 \item $P_T$ is directly reachable from $P_A$;
		 \item $P_T$ is on the boundary of OBST: $P_T \in \partial \text{OBST}$; 
		 \item For some $\varepsilon>0$ small enough and all $\lambda \in (1,1+\varepsilon)$, $P_T(\lambda)$ is directly reachable from $P_A$ and  $P_T(\lambda) \notin \partial OBST$.
		\end{enumerate} 
		\ds{In addition, if $\exists \mathcal{O}$ such that $P_A \in \partial \mathcal{O}$. Let $P_A^+(\lambda)$ and $P_A^-(\lambda)$ be the obstacle boundary tangents of $P_A$ via Definition} \ref{def:OBTangents}. The point $P_T$ is an FTP if:
		 \begin{enumerate}[label=(\alph*)]
		 \item  
		 \ds{There exists $\varepsilon_1>0$ such that for all $\lambda \in (0,\varepsilon_1)$ we have $P_A^+(\lambda) \in FREE$.} 
		 \item \ds{$P_T$ is directly reachable from $P_A$}
		 \item \ds{$P_T$ equals $P_A^+\left( \max(\inf_{\lambda>0} \{ \lambda: P_A^+(\lambda) \notin \partial OBST\},\varepsilon_2) \right)$, for some $\varepsilon_2>0$ small enough}
		 \end{enumerate}
		 Or if:
		 \begin{enumerate}[label=(\alph*)]
		 \item  
		 \ds{There exists $\varepsilon_1>0$ such that for all $\lambda \in (-\varepsilon_1,0)$ we have $P_A^-(\lambda) \in FREE$.}
    		 \item \ds{$P_T$ is directly reachable from $P_A$}
		     \item \ds{$P_T$ equals $P_A^- \left( \min(\sup_{\lambda<0} \{ \lambda: P_A^-(\lambda) \notin \partial OBST\},-\varepsilon_2) \right)$, for some $\varepsilon_2>0$ small enough}
		 \end{enumerate}
		 
\end{definition}

As illustrated in Fig.\ \ref{fig:TangentDefinitionExample}, the FTP on the left is located at the end of the obstacle rim, in accordance with the definition of an FTP which requires that $P(\lambda) \notin \partial OBST$ for $\lambda \in (1,1+\varepsilon)$.
\begin{figure}[H]
	\centering
	\includegraphics[width=5.6cm]{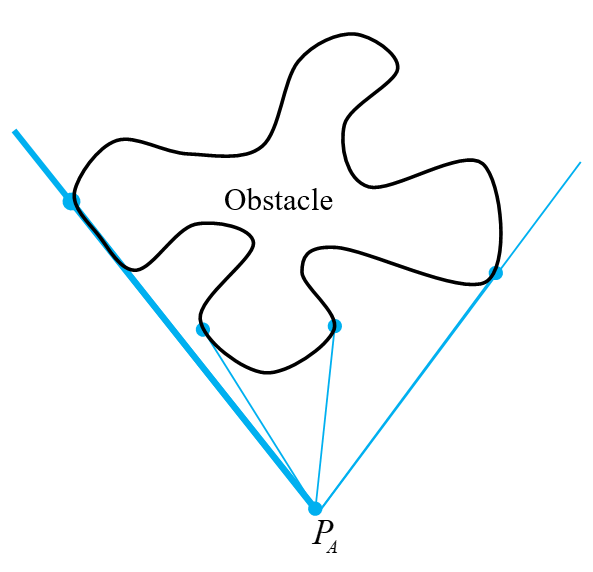}
	\caption{FTP Example: An obstacle with four FTPs from $P_A$}\label{fig:TangentDefinitionExample}
	\par
\end{figure}
In Fig.\ \ref{fig:PAonOBST} \ds{the starting points $P_{A1}$ and $P_{A2}$ reside on the boundary of an obstacle. In this example $P_{A1}$ has only two FTPs and $P_{A2}$ has three FTPs. Note that the obstacle boundary tangents of $P_{A1}$ are inside the obstacle; namely, there is no $\varepsilon_1>0$ such that $P_{A1}^+(\lambda) \in FREE$ for all $\lambda \in (0,\varepsilon_1)$ and $P_{A1}^-(\lambda) \in FREE$ for all $\lambda \in (-\varepsilon_1,0)$. However, for $P_{A2}$ the obstacle boundary tangents are outside the obstacle which results in two additional FTPs, $P_{T-}$ and $P_{T+}$. From each FTP the aircraft will require to advance along the obstacle boundary in an iterative manner until $P_B$ becomes reachable or another obstacle becomes induced by the descent.}
\begin{figure}[H]
	\centering
	\includegraphics[width=5.6cm]{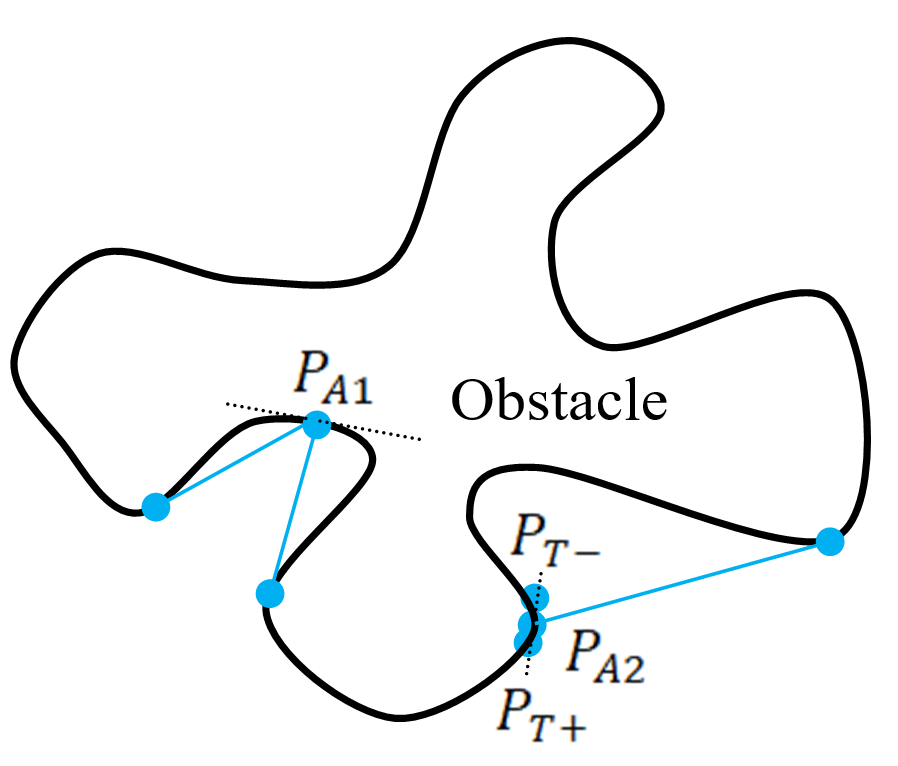}
	\caption{FTP Example: the starting point $P_A$ is on the obstacle boundary}\label{fig:PAonOBST}
	\par
\end{figure}

Consider now an ALO trajectory from $(P_A,Z_A)$ to $P_B$. If $P_B$ is directly reachable from $P_A$, then clearly, by Theorem \ref{thrm:Theorem1}, the ALO trajectory is simply a fixed heading glide to $P_B$. Otherwise, if an obstacle stands in our way, the following holds.
\begin{theorem}\label{thrm:TangentToObstacle}
	Suppose $P_B$ is not directly reachable from $P_A$ and that the set of FTP points from $P_A$ is finite. Then any ALO trajectory from $(P_A,Z_A)$ to $P_B$ must include a fixed-heading glide segment from $P_A$ to an FTP.
\end{theorem}
\begin{proof}
Denote set of FTPs by $\{P_{i} \}_{i=1}^N$, and let $P_{i}(\lambda)$ be the convex combination of $P_{i}$ to $P_A$, as per Definition \ref{def:CCP}. Extend the line segment from $P_A$ to $P_{i}$ until it touches the next obstacle, and denote the additional segment by $E_i$ (see Fig.\ \ref{fig:DRR_Boundary}). Formally, let 
$$
\lambda^*_i 
= \inf_{\lambda>1}\{\lambda : P_{i}(\lambda)\in \partial OBST \}
$$  
and
$$ 
E_i = \{P_{i}(\lambda) : 1\leq \lambda \leq \lambda_i^* \} 
$$

Let us divide the entire free space $\text{FREE}$ into two sets: The set $\text{F}_1$ of directly reachable points from $P_A$ (Definition \ref{def:DirectReach}), and the set $\text{F}_2$ of \textit{potentially reachable} points, namely, those points in FREE which are not directly reachable from $P_A$.
Then, as illustrated in Fig.\ \ref{fig:DRR_Boundary}, the segments 
$\{E_i \}_{i=1}^N$ serve as the boundary between the directly reachable set $\text{F}_1$ and the potentially reachable set $\text{F}_2$.

Evidently $P_A$ is in the reachable set (from itself), while $P_B$ is not by the Theorem assumption. Therefore, the ALO trajectory (like any feasible trajectory) must cross at least one of these segments $E_i$ at some \hlhl{point} $e_i=P_{i}(\lambda_i)$, where $1\leq \lambda_i \leq\lambda_i^*$.

As $e_i$ is directly reachable, the ALO trajectory to $e_i$ is the fixed heading trajectory, which by definition of $E_i$ must pass though $P_{i}$.  
\end{proof}
\begin{figure}[H]
	\centering
	\includegraphics[width=5.2cm]{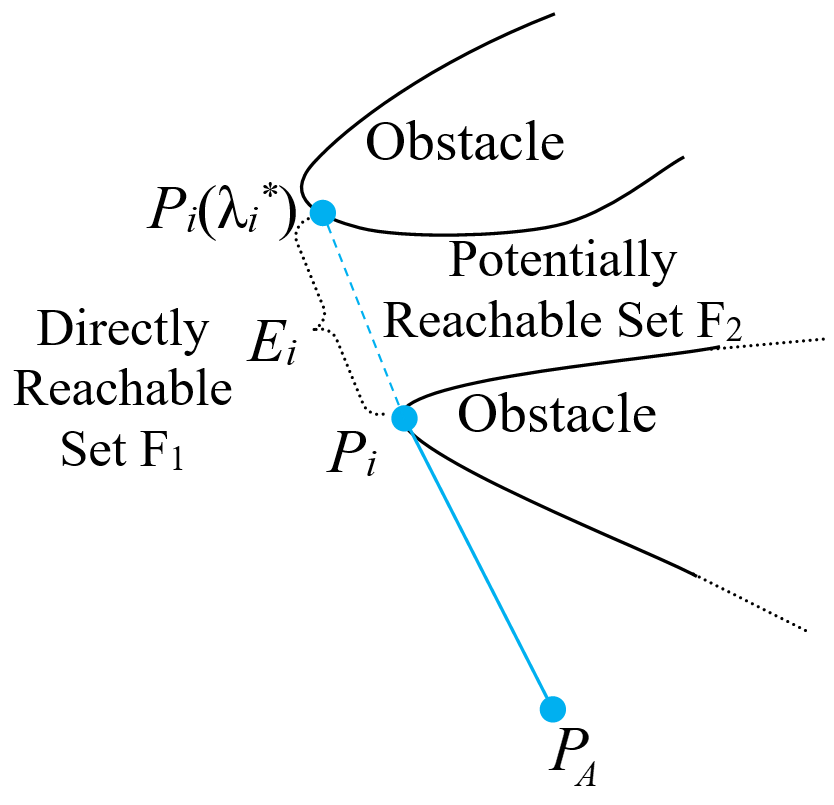}
	\caption{The boundary between the directly reachable and potentially reachable sets}\label{fig:DRR_Boundary}
	\par
\end{figure}
As illustrated in Fig.\ \ref{fig:AltitudeLossOptimalPathThroughObstacle}, to reach $P_B$, trajectories 1, 2 and 3 pass through points $\{P_{i}(\lambda_i)\}_{i=1}^3$ respectively. As they do not include the fixed heading trajectories from $P_A$ to $P_{i}$ they are sub-optimal.
\begin{figure}[H]
	\centering
	\includegraphics[width=8.2cm]{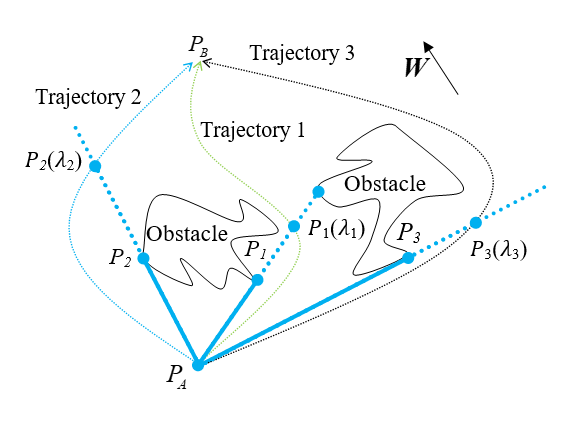}
	\caption{Three candidate trajectories for the ALO paths illustrated on a local obstacle map}\label{fig:AltitudeLossOptimalPathThroughObstacle}
	\par
\end{figure}
Based on Theorem \ref{thrm:TangentToObstacle}, we can derive an iterative graph search algorithm to obtain the ALO trajectory. A concrete algorithm, based on the standard A$^*$ search scheme, is presented in Subsection \ref{subsecc:TheAloAlg}. 
Essentially, the algorithm starts at the initial point $(P_A,Z_A)$, \emph{expands} this point by constructing the local obstacle map from this position, and finding the respective FTPs which serve as the successor nodes in the search graph. The next node to be expanded is chosen by the key or ranking function of the A$^*$ algorithm, and the process continues iteratively until an optimal path to the target is found and verified.  

The computational complexity of the outlined procedure clearly depends on the number of FTPs that need to be explored per obstacle. We proceed to show that this number can be reduced to two, even for non-convex obstacles. 

Let us first relate to the case in which the destination, $P_B$, is outside the convex hull of an obstacle. In this case, it should be intuitively clear that we can explore only the two most extreme FTPs and not the entire FTP set, as illustrated in Fig.\ \ref{fig:ConvexHullExample}. 
\begin{figure}[H]
	\centering
	\includegraphics[width=5.2cm]{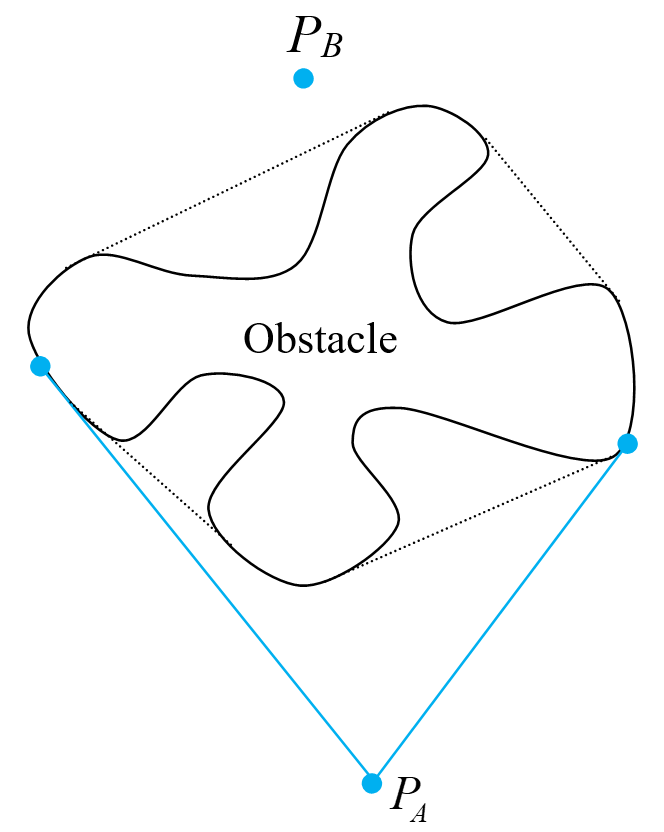}
	\caption{Illustration of a destination outside the obstacle convex hull}\label{fig:ConvexHullExample}
	\par
\end{figure}
The latter observation indeed follows as a special case of Theorem \ref{thrm:TangentCountToObstacle} below. A less immediate case is when the target is within the convex hull of an obstacle, as illustrated in Fig.\ \ref{fig:FtpCountIllustration}. Our statement requires the following definition.
\begin{definition}[Essential FTPs]
\label{def:EssentialFTP}
	Consider a point $P_A$ and its FTP set $\{P_{i} \}_{i=1}^{N}$ with respect to an obstacle $\mathcal{O}$. Suppose that $\{P_{i} \}_{i=1}^{N}$ is arranged in a monotonously increasing order in terms of heading from $P_A$ to $P_{i}$. 
	
	Let $C_i$, $i=1...N$, be the area enclosed between $P_i$, $\partial \mathcal{O}$ and $P_{i+1}$ (see Fig.\ \ref{fig:FtpCountIllustration}). Identify $P_{N+1}$ with $P_1$. 
	The essential FTPs of $\mathcal{O}$ with respect to $P_A$ are the pair $\{P_{j}, P_{j+1}\}$ such that $P_B \in C_j$. 
	\end{definition}
	For example, in Fig.\ \ref{fig:FtpCountIllustration} the essential FTPs are $\{P_1,P_2\}$.
\begin{theorem}\label{thrm:TangentCountToObstacle}
Suppose $P_B$ is not directly reachable from $P_A$. Then any ALO trajectory from $(P_A,Z_A)$ to $P_B$ must include one of the fixed heading glide segments from $P_A$ to the two essential FTPs.
\end{theorem}
\begin{proof}
	 Let us observe the essential FTPs $\{P_{j}, P_{j+1}\}$ as illustrated in Fig.\ \ref{fig:FtpCountIllustration}. According to Theorem \ref{thrm:TangentToObstacle} the ALO trajectory must include a fixed heading segment from $P_A$ to some FTP. Any trajectory from $P_A$ to an FTP, $P_{k}$ which is not an essential FTP, should cross the fixed heading segments from $P_A$ to $P_{j}$ or from $P_A$ to $P_{j+1}$ at some \hlhl{point} $Q$. The trajectory from $P_A$ through $P_{k}$ to $Q$, unlike the trajectory from $P_A$ directly to $Q$, is not a fixed heading trajectory, and thus via Theorem \ref{thrm:Theorem1} it is not ALO.
	 \end{proof}
\begin{figure}[H]
	\centering
	\includegraphics[width=6.2cm]{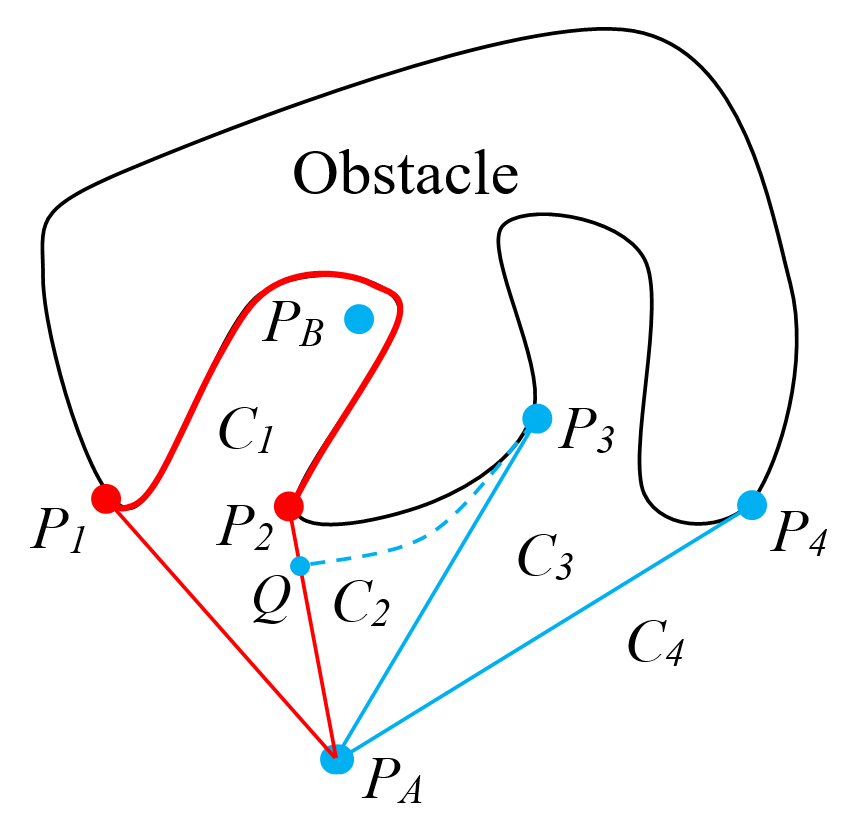}
	\caption{Essential FTPs illustration of the four FTPs, $P_1$ and $P_2$ are the essential pair}\label{fig:FtpCountIllustration}
	\par
\end{figure}

\subsection{From Continuous to Discrete Terrain Map}\label{subsec:ToDiscrete}
In this subsection, we assume that the terrain map is as a DTM given over a discrete grid, hence the obstacles set is obtained at each stage over that grid. We propose two alternative approaches to perform interpolation of the discrete obstacles set to the continuous domain. Each alternative defines what is a feasible ("safe") path. The DTM is represented as an $M \text{ by } N$ matrix with discrete samples of the elevation data in a resolution of $\Delta X$ and $\Delta Y$ m. To obtain the discrete local obstacle map, we sample the local obstacles function, $\textit{LO}$, for all $1 \leq m \leq M,1 \leq n \leq N$ as the matrix:
\begin{equation}\label{eq:DiscriteLocalObstacle}
\textit{LO}[m,n] = \textit{LO}(X_{I0}+(m-1)\Delta X,Y_{I0}+(n-1)\Delta Y;P_A,Z_A)
\end{equation}
where $\textit{LO}(X,Y;P_A,Z_A)$ is given by Eq.\ (\ref{eq:LO}), and $(X_{I0},Y_{I0})$ is the \hlhl{point} of the elevation sample $\textit{DTM}[1,1]$. Due to the discretization the boundary of the obstacles $\textit{LO}(X,Y;P_A,Z_A)=0$ is not generally sampled; therefore, we provide the following two interpolation schemes.

\emph{Approach 1.\ Linear approximation over a triangulation:}
in this approach, we form a standard triangulation by dividing each square into two triangles as illustrated in Fig.\ \ref{fig:Trangulation}. The obstacle boundaries are then identified using linear interpolation over each square. This will give polygonal obstacles, whose vertices lie on the sides of the triangles (at most two per triangle). Since only vertices can be FTPs, we can construct a finite algorithm that implements the exact (continuous) one by employing Theorem \ref{thrm:TangentToObstacle} and Theorem \ref{thrm:TangentCountToObstacle}.
\begin{figure}[H]
	\centering
	\includegraphics[width=7.2cm]{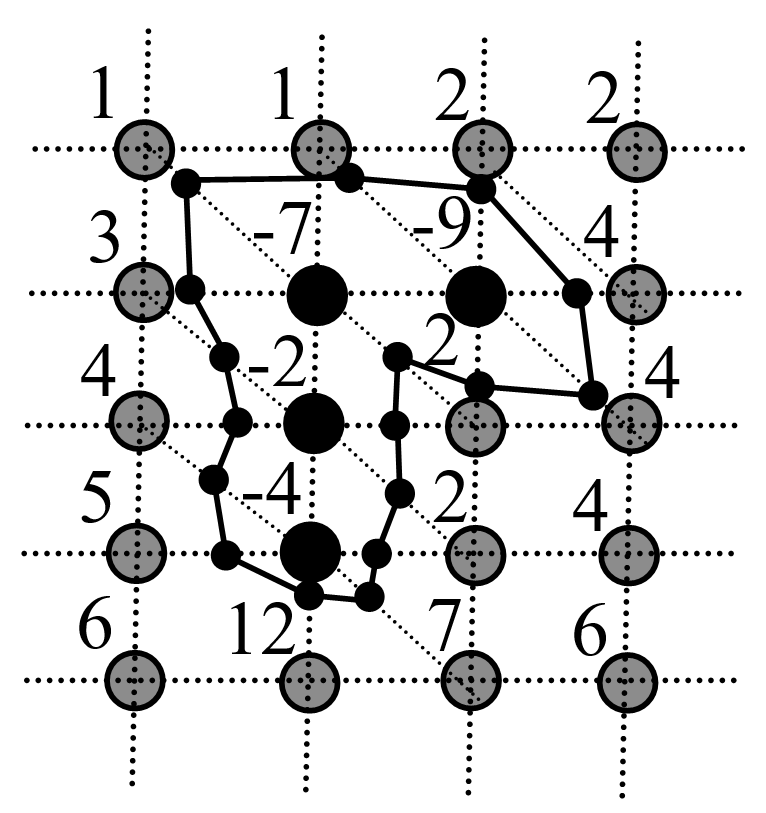}
	\caption{Local obstacle map for the linear approximation approach}\label{fig:Trangulation}
	\par
\end{figure}

\emph{Approach 2. Safe squares:}
This approach is somewhat more conservative.
Define a safe square as one for which all its vertices maintain $\textit{LO}[m,n] \ge 0$,  as illustrated in Fig.\ \ref{fig:SafeSquares}.
A square is \emph{unsafe} if at least one of its vertices has $\textit{LO}[m,n] < 0$. 
The obstacles set is defined as the union of the unsafe squares, while FREE is the complement of the obstacles set. We include in FREE its boundary.
A feasible (safe) path must be in FREE.
\begin{figure}[H]
	\centering
	\includegraphics[width=7.2cm]{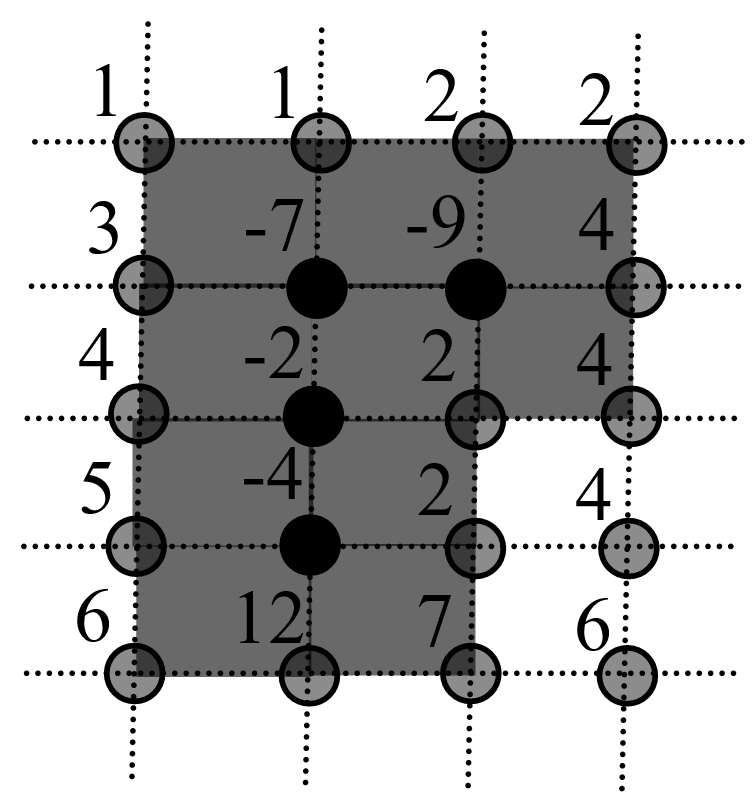}
	\caption{Local obstacle map for the safe squares approach}\label{fig:SafeSquares}
	\par
\end{figure}
In this approach the formed obstacles (the complement of FREE) are polygonal with vertices on the grid points; thus, the FTPs are also on grid points only. 

An FTP point, via its definition, should be directly reachable from the current position. Evaluating direct reachability with approach 2 is simpler, as the obstacles encompass entire squares only; therefore, the safe path should not intersect squares defined as obstacles. 
As this approach allows for a simpler implementation, we henceforth focus on it in the description of our algorithm.

Having defined the local obstacle map as polygonal sets in continuous space we can apply the results of Subsection \ref{subsec:TrajOptAlgorithm} to obtain a finite algorithm to find the FTPs at each stage of the overall algorithm.

\subsection{The ALO Trajectory Planning Algorithm}
\label{subsecc:TheAloAlg}
In this subsection, we first define the graph on which a trajectory search algorithm can find the ALO trajectory. Then, we choose the A$^*$ shortest path algorithm to obtain the ALO trajectory due to its efficiency property. In the sequel, we derive the algorithm pseudocode. 

Let us create the graph, $G=(V_r,E)$, with vertices, $V_r$, and edges, $E$. The graph is created by employing an iterative algorithm. Each vertex $V_i=(P_i,d_i) \in V_r$ consists of the location point $P_i$ and the altitude loss, $d_i$, on a trajectory to $V_i$. The algorithm initializes the graph and starts exploration from the engine cutoff vertex $V_A=(P_A,0)$ and continues through intermediate FTP vertices. At each explored node $V_i=(P_i,d_i)$ we add the vertex $V_j=(P_B,d_j)$ if $P_B$ is directly reachable from $V_i$ otherwise we add, $\{V_j=(P_j,d_j)\}$, the set of FTPs from $V_i$, to the graph vertices set, $V_r$. Also, we add the all of the edges $\{(V_i,V_j)\}$ to the edges set $E$. In case $V_j$ is an FTP we add it to the set of nodes to be explored in the next iterations. The algorithm continues until the set of nodes to be explored is empty.

Now let us show that the ALO trajectory can be obtained from the graph, $G=(V_r,E)$.
\begin{theorem}\label{thrm:ALO_trajectoryExistence}
The ALO trajectory from position $(P_A,Z_A)$ to point $P_B$ can be obtained from the graph $G=(V_r,E)$.
\end{theorem}
\begin{proof}
Let us prove by induction the existence in $G$ of the ALO trajectory vertices and edges. The first ALO trajectory vertex, $V_{i=1}$, is $V_A=(P_A,0)$; it exists in the graph since the graph build algorithm initializes $G$ with $V_A$. At the first iteration, the vertex $V_{i=1}=V_A$ is explored. In case $P_B$ is directly reachable from $V_{i=1}$, than $V_{i=2}$ is $(P_B,d_2)$; thus, the trajectory from $P_A$ directly to $P_B$ exists in the graph $G$ which is ALO via Theorem \ref{thrm:Theorem1}. In case $P_B$ is not directly reachable from $V_{i=1}$ than the adjacent vertices set, $\{V_j\}$, is the set of FTPs from $V_A$. Via Theorem \ref{thrm:TangentToObstacle} at least one edge from $P_A$ to an FTP is part of the ALO trajectory; thus, in both cases, the second vertex, $V_{i=2}$, and the edge $(V_{i=1},V_{i=2})$ of the ALO trajectory exist on graph $G$. 

Assuming vertex number $i$ of the ALO trajectory, $V_i$, exists in $G$, let us show that the vertex $V_{i+1}$ and edge $(V_i,V_{i+1})$ of the ALO trajectory exist in $G$. In case $P_B$ is directly reachable from $V_i=(P_i,d_i)$ than the vertex $V_j=(P_B,d_j)$ and edge $(V_i,V_j)$ exist in $G$; thus, the trajectory from $P_i$ directly to $P_B$ exists in the graph $G$ (which is ALO via Theorem \ref{thrm:Theorem1}). In case $P_B$ is not directly reachable from $(P_i,d_i)$, 
than the adjacent vertices of $V_i$, $\{V_j=(P_j,d_j)\}$, are the FTPs from $V_i$. Via Theorem \ref{thrm:TangentToObstacle} at least one edge from $(P_i,d_i)$ to the FTP $V_j$ is part of the ALO trajectory; thus, vertex $V_{i+1}=V_j$ and edge $(V_i,V_{i+1})$ of the ALO trajectory also exist in graph $G$. Therefore, the ALO trajectory can be obtained from graph $G$.
\end{proof}
Next, we present an outline of the chosen path planning algorithm, in a pseudocode format. The algorithm returns an optimal (ALO) path from the initial position and altitude $(P_A,Z_A)$ to the candidate landing site $P_B$, or reports a failure if no such obstacle-avoiding path exists.   
The standard A$^*$ algorithm \cite{ref:Hart68,ref:Nilsson80} is used to guide the graph construction and search. The heuristic function $h(V_i)$, for some vertex $V_i=(P_i,d_i)$, is naturally taken as the minimal straight-glide altitude loss from position $(P_i,Z_A+d_i)$ to the target $P_B$. This heuristics is admissible (optimistic), which guarantees the optimality of A$^*$, and also consistent (monotone), which entails certain efficiency properties.
A heuristic function is said to be admissible if it never overestimates the cost of reaching the goal, i.e. the cost it estimates to reach the goal is not higher than the lowest possible cost from the current point in the path.

The heuristic function employed in Algorithm \ref{alg:3DOptimalTrajectorySearch}, $h(V_i)$, for a vertex $V_i=(P_i,d_i)$, is the minimal altitude loss in free space from point $P_i$ to the landing site $P_B$. This heuristic is a consistent underestimate of the minimal altitude loss as via Theorem \ref{thrm:Theorem1} a fixed heading trajectory is ALO in free space; therefore, the ALO trajectory in the presence of ground obstacles may not maintain a fixed heading and thus must incur greater altitude loss.
Thus, the proposed A$^*$ variant yields the optimal trajectory on the graph whose nodes are the engine cutoff location $(P_A,0)$ the landing site $(P_B,d_j)$ and the FTPs.

The proposed ALO algorithm is presented in Algorithms \ref{alg:3DOptimalTrajectorySearch}. The algorithm returns either the ALO path from $P_A$ to $P_B$, or {\it failure} if none exits. 
Algorithm \ref{alg:3DOptimalTrajectorySearch} is essentially the standard A$^*$ algorithm, following \cite{ref:Nilsson80}, with the straight-path optimal altitude loss serving as the link costs.   
Algorithm \ref{alg:functions} outlines the relevant application-specific functions. \hlhl{With standard on-board computer capability, the sparsity feature of our unique variant yields calculation cycles of less than three seconds. 
And, of course, the engineering application phase of our derivation will include Monte Carlo studies,
addressing the uncertainties issue. However, it a priori assures non-divergence, penalizing reachability as a function of uncertainty levels. In the sequel, the pseudo-codes are presented:}
\begin{algorithm}[H]
	\caption{Altitude-Loss-Optimal Trajectory Search with A$^*$}
	\label{alg:3DOptimalTrajectorySearch}
	\begin{algorithmic}[1]
		\Function{ALO-Trajectory-Search}{$P_A ,Z_A, P_B, W, \textit{DTM}$}	\label{a:l1}
		\State $V_A \gets (P_A,0)$
		\State $g(V_A) \gets 0$; $h(V_A) \gets $\Call{ALT-Loss}{$P_A, P_B$};		\label{a:l2}
		\State $f(V_A) \gets g(V_A)+h(V_A)$; $\parent(V_A) \gets \textit{nill}$ 			\label{a:l3}
		\State $\OPEN \gets$ a list ordered by $f(\cdot)$, with $V_A=(P_A,0)$ as the initial element	\label{a:l4}
		\State $\CLOSED \gets \emptyset$  (an empty list) 						\label{a:l5}
		\Loop	\label{a:l6}
		\If {$\OPEN=\emptyset$} \Return{failure} \EndIf 						\label{a:l7}
		\State $V_i=(P_i,d_i) \gets$ element $V_i$ of $\OPEN$ with the smallest value of $f(V_i)$ 			\label{a:l8}
		\State Add $V_i$ to $\CLOSED$ and  remove $V_i$ from $\OPEN$  		\label{a:l9}
		\If{$Z_A + f(V_i) > -\textit{DTM}(P_B)$} \Return{failure} \EndIf 			\label{a:l10}
		\If{$P = P_B$}  \Return{($g(V_i)$,\Call{Optimal-Path}{$V_A,V_i$})} \EndIf 		\label{a:l11}
		\State $\successors \gets$  \Call{Expand}{$P_i,d_i+Z_A$} (find the adjacent vertices of $V_i$)  					\label{a:l12}
		\For{every $P_j \in \successors$ } 				\label{a:l13}
			\State $g_j \gets g(V_i) + \Call{Alt-Loss}{P_i,P_j}$ (note that $g(V_i)=d_i$)
			\State $V_j \gets (P_j,g_j)$
			\If{$V_j \notin \CLOSED$}
    			\State $h(V_j) \gets \Call{Alt-Loss}{P_j,P_B}$ \label{a:l14}
    			
    			\If{($V_j\in\OPEN  \And  g_j < g(V_j)$)} 							\label{a:l15}
    				\State  $g(V_j) \gets g_j$; $f(V_j) \gets g(V_j)+h(V_j)$; $\parent(V_j) \gets V_i$ \EndIf 	\label{a:l16}
    			\If{$V_j\not\in\OPEN$}  add $V_j$ to $\OPEN$  					\label{a:l17}
    				\State  $g(V_j) \gets g_j$; $f(V_j) \gets g(V_j)+h(V_j)$; $\parent(V_j) \gets V_i$ \EndIf  	\label{a:l18}
		    \EndIf
		\EndFor
		\EndLoop
		\EndFunction
		\State  \label{a:l19}
		\Function{Optimal-Path}{$V_A,V_B$} 				\label{a:l20}
    			\State  $\pathh \gets [V_B]$; $V_i \gets V_B$ 	\label{a:l21}
                \While{$V_i \not = V_A$}				\label{a:l22}
				\State $V_i \gets \parent(V_i)$			\label{a:l23}
                	\State $\pathh \gets [V_i,\pathh]$		\label{a:l24}
			\EndWhile
			\State \Return{\pathh}				\label{a:l25}
		\EndFunction

		\algstore{3Dmyalg}	
	\end{algorithmic}
\end{algorithm}
The inputs to the algorithm in line \ref{a:l1} are the initial position $(P_A,Z_A)$, the target (landing site) point $P_B$, the wind components, $W=(W_{X},W_{Y})$, and the digital terrain map, DTM. 
For simplicity, we identify nodes in the search graph with their location point $P$ and altitude loss $d$. 
The basic quantities assigned to a node $V_i=(P_i,d_i)$ are $g(V_i)$, the altitude loss of the best path discovered so far from $P_A$ to $V_i$; the heuristic function $h(V_i)$, an under-estimate of the altitude loss from $V_i$ to $P_B$; and their sum $f(V_i)$ which serves as an estimate for the total altitude loss for a path that goes through $P_i$. Also, $\parent(V_i)$ identifies the predecessor to $V_i$ for tracing  the optimal path. A $\CLOSED$ list contains nodes whose minimal altitude loss has been determined, and an $\OPEN$ list contains {\it frontier} nodes that are waiting to be explored.  The $\OPEN$ list is typically implemented as a priority queue, with key $f$.  
Lines \ref{a:l2}-\ref{a:l3} initialize the search graph with the initial node $(P_A,0)$, and lines \ref{a:l4}-\ref{a:l5} initialize the $\OPEN$ and $\CLOSED$ sets. The body of the algorithm is a loop that terminates with an optimal path or a failure. In line \ref{a:l7}, \textit{failure} is declared if there are no more nodes in $\OPEN$ to explore (while $P_B$ has not been reached before). 
In lines \ref{a:l8}-\ref{a:l9}, the next node to be expanded is chosen as the one in $\OPEN$ with minimal key-value $f(V_i)$, and moved $\CLOSED$.
Line \ref{a:l10} (the  only non-standard addition) checks if the current best under-estimate in $f(V_i)$ still allows to reach the target above its ground level; if not, it returns \textit{failure}. This additional check ensures that the search does not continue in vain even if $\OPEN$ is not empty. Line \ref{a:l11} terminates the algorithm with success if a feasible path to $P_B$ with minimal altitude loss has been determined. It then returns the minimal altitude loss $g((P_B,d))$, and the optimal path which is traced back via the \textsc{Optimal-Path} function in lines \ref{a:l20}-\ref{a:l25}. Next, line \ref{a:l12} finds the successors to node $V_i$ via the function \textsc{Expand}, which receives both the node \hlhl{projection onto the 2D horizontal plane} and altitude. Finally, lines \ref{a:l13}-\ref{a:l19} update, for each successor $V_j$ which is not already in $\CLOSED$, the altitude loss $g(V_j)$ of the best path found so far to $V_j$, and its parent node in that path. 

Some explanations for the functions in Algorithm \ref{alg:functions} are interleaved as comments.  
As noted, pre-calculation of the ALO manifold is possible and useful for computational efficiency.
The two functions in lines \ref{a:l34} and \ref{a:l36} of \textsc{Expand} are not explicitly specified.   
\textsc{Directly-Reachable} simply checks if the straight-line path from $(P,d)$ to $P_B$ lies in the 
FREE part of the obstacles map \textit{LOMap}, and is a standard procedure in computational geometry. 
The \textsc{Find-Extreme-FTPs} function relies on the results of Subsection \ref{subsec:TrajOptAlgorithm} and the characterizations in Definitions \ref{def:FTP} and \ref{def:EssentialFTP},
and can be implemented by directly following these definitions. An efficient implementation however requires 
more advanced methods from computational geometry, and is outside the scope of the present paper.
Some related algorithms may be found in \cite{ref:Liu94}.
\begin{algorithm}[H]
\caption{Functions for Algorithm \ref{alg:3DOptimalTrajectorySearch}}
\label{alg:functions}
	\begin{algorithmic}[1]
		\algrestore{3Dmyalg}			

		\Function{ALO-Manifold}{$\Delta P$} 		\label{a:l26}
		\LeftComment{The ALO manifold function $M(x,y)$ is specified in Eq.\ \ref{eq:TheALOManfiold}}
           \LeftComment{For efficiency $M(\cdot)$ is pre-calculated on suitably dense grid}
           \State $\Delta Z \gets  M(\Delta P)$					\label{a:l27}					
           \State \Return{$\Delta Z$}					\label{a:l28}
		\EndFunction

		\Function{Alt-Loss}{$P1,P2$} 			\label{a:l29}
		\LeftComment{Altitude loss from $P_1$ to $P_2$ for optimal straight-glide with wind}
		\State $\Delta Z \gets \textsc{ALO-Manifold}(P2-P1)$		\label{a:l30}
 		\State \Return{$\Delta Z$}		\label{a:l31}
		\EndFunction

		\Function{Expand}{$P$, $Z_P$}  \label{a:l32}
		\LeftComment{Obtain the essential successors of node $P$}
		\State $\textit{LOMap} \gets$ \Call{Calculate-Local-Obstacle-Map}{$P$, $Z_P$} \label{a:l33}
 		\If{\Call{Directly-Reachable}{$P,P_B,\LOMap$}} \label{a:l34}
		\State  $\successors \gets \{P_B\}$                          \label{a:l35}
		\Else \ \  $\successors \gets $ \Call{Find-Extreme-FTPs}{$P,\LOMap$} \label{a:l36}
		\EndIf
	      \State \Return{$\successors$}    \label{a:l37}
		\EndFunction

		\Function{Calculate-Local-Obstacle-Map}{$P$, $Z_P$} 		\label{a:l38}
		\LeftComment{Calculate the local obstacle map as seen from $P$, for the given DTM}
           \LeftComment{We emply the safe squares approach, as per Section \ref{subsec:ToDiscrete}}
           \State Compute $\LO$ on grid points centered at $P$, Eqs.\ (\ref{eq:LO}),(\ref{eq:DiscriteLocalObstacle}) \label{a:l39}
		\State  $\LOMap \gets$ Mark all map squares as OBST or FREE           \label{a:l40}
		\State \Return{$\LOMap$}  \label{a:l41}
		\EndFunction

%
	\end{algorithmic}
\end{algorithm}	

Now we proceed to show that the proposed algorithm obtains the ALO trajectory from graph $G$.
\begin{theorem}\label{thrm:AlgoObtainsTheALoTraj}
Algorithm \ref{alg:3DOptimalTrajectorySearch} obtains the ALO trajectory from position $(P_A,Z_A)$ to $P_B$.
\end{theorem}
\begin{proof}
The edges of every explored node $V_i=(P_i,d_i)$ to an adjacent node $V_j=(P_j,d_j)$ are calculated by employing the \textsc{Expand} method to obtain $P_j$ and calculating $d_j$ by adding the minimal altitude loss between $P_i$ and $P_j$ to $d_i$. In case $P_B$ is directly reachable from $V_i$ than $P_j=P_B$, otherwise $P_j$ is an FTP. Therefore, the explored edges and vertices are those of graph $G=(V_r,E)$, i.e., Algorithm \ref{alg:3DOptimalTrajectorySearch} (the A$^*$ algorithm) operates on graph $G$.

The heuristic function, $h(V_i)$, employed in Algorithm \ref{alg:3DOptimalTrajectorySearch} is the minimal altitude loss in free space from $V_i=(P_i,d_i)$ to $P_B$. This function via Theorem \ref{thrm:Theorem1} yields the minimal cost in free space $P_B$ and therefore serves as a lower bound in case $P_B$ is not directly reachable due to obstacles. Therefore, $h(V_i)$ is a consistent underestimate of the minimal altitude loss from $V_i$ to $P_B$.

Now, via \cite[Result 4  pp.\ 78]{ref:Nilsson80} we have that A$^*$ finds the optimal trajectory in graph $G$, given $h(V_i)$ is an admissible consistent heuristic. As via Theorem \ref{thrm:ALO_trajectoryExistence} the ALO trajectory exists in graph $G$, the obtained trajectory is the global ALO trajectory from $V_A=(P_A,0)$ to $P_B$.
\end{proof}

\hlhl{Emphasizing the inherent global optimality of our Optimization formulation: 
(a) the math model of the dynamics/aerodynamics of our system, incorporating the wind effect, results in cone-like manifolds, periodically updated on the fly.
(b) contours, obtained via intersections between these envelopes and the terrain, constitute inputs to the algorithm which on-line finds the tangent grid nodes, (c) over these nodes and the candidate landing strips, our accelerated algorithm propagates the trajectory towards these destinations.}
Now, we proceed to show that the search graph used in Algorithm \ref{alg:3DOptimalTrajectorySearch} is finite. First, note that our algorithm extracts obstacles as 2D polygons, due to that we can show that amount of FTPs derived from a local obstacle map is finite.
\begin{proposition}\label{prp:PolygonalFTP}
Given a polygonal local obstacles nodes set $\{(x_i,y_i)\}_{i=1}^N$, the FTPs, $\{P_j\}_{j=1}^K$, of the local obstacles is a subset of the polygonal nodes set; namely, $\{P_j\}_{j=1}^K \subset \{(x_i,y_i)\}_{i=1}^N$ .
\end{proposition}
\begin{proof}
According to the FTP Definition \ref{def:FTP} (b) we have $\{P_j\}_{j=1}^K \in \partial \text{OBST}$. Therefore, FTPs are either on (i) from the set $\{(x_i,y_i)\}_{i=1}^N$ or (ii) on $\partial \text{OBST}$  not including $\{(x_i,y_i)\}_{i=1}^N$. In case (ii) we have that either condition (a) or condition (c) of Definition \ref{def:FTP} does not hold (see Fig.\ \ref{fig:FTPsToPolygon}). Therefore, the FTPs can only be of case (i), i.e., all FTPs are from the set $\{(x_i,y_i)\}_{i=1}^N$.
\end{proof}
\begin{figure}[H]
	\centering
	\includegraphics[width=5.2cm]{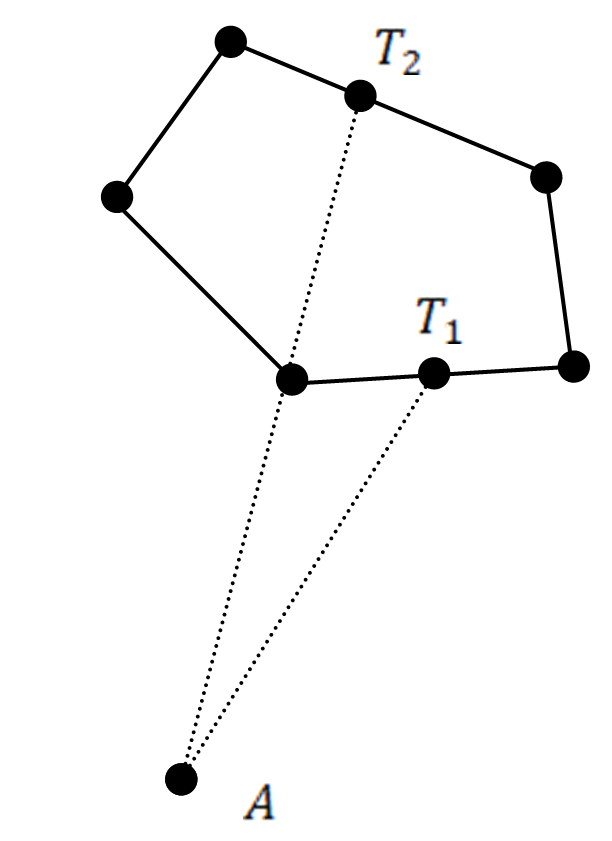}
	\caption{Illustration of a polygonal local obstacle. The points $T_1$ and $T_2$ represents a contradiction to Definition \ref{def:FTP} (c) and Definition \ref{def:FTP} (a) respectively }\label{fig:FTPsToPolygon}
	\par
\end{figure}
\begin{proposition}\label{prp:finiteGraph}
Given an $M$ by $N$ discrete elevation map. The vertices set of the graph in Algorithm \ref{alg:3DOptimalTrajectorySearch} is finite.  
\end{proposition}
Let us use the safe squares representation of local obstacles mapping (Section \ref{subsec:ToDiscrete} above) as means for substantiating the following proof.
\begin{proof}[Proof of Proposition \ref{prp:finiteGraph}]
 In the safe squares approach the local obstacles are polygons with vertices as a subset of the grid nodes of an $M$ by $N$ discrete elevation map. Thus, via Proposition \ref{prp:PolygonalFTP} the FTPs are also a subset of the $M$ by $N$ discrete elevation map grid. Due to that the set of horizontal nodes is finite and bounded by $M \cdot N$ nodes. Each time the algorithm transverses from one horizontal node to a subsequent one it will lose at least the minimal altitude along a single discrete elevation map grid. Thus, in case the algorithm revisits the same horizontal point from the $M \times N$ grid it will reach the ground level after enough visits. Therefore, the graph spanned from FTPs of local obstacles the engine cutoff location, and the landing site location is finite.
\end{proof}
Since the search graph is finite, Algorithm \ref{alg:3DOptimalTrajectorySearch} will converge, via Theorem \ref{thrm:AlgoObtainsTheALoTraj}, to the ALO trajectory from position $(P_A,Z_A)$ to $P_B$.

A completed search graph is illustrated in Fig.\ \ref{fig:ResultingGraphFollowingOGS}.
The initial node is the engine cutoff point $P_A$, the target node is the intended landing site $P_B$, and the remaining nodes are FTPs that were computed iteratively as part of the algorithm. The link cost (or weight) is the optimal straight-glide altitude loss between the node positions, effectively given by the ALO manifold computed above. 
Each node position $V_i=(P_i,d_i)$ is adjoined as part of the algorithm with a cost $g(V_i)$ which holds the altitude loss from $P_A$ to that point, the heuristic $h(V_i)$ which is an under-estimate of the optimal altitude loss from $V_i$ to $P_B$, and the sum $f(V_i)=g(V_i)+h(V_i)$ which is used to choose the next node to be explored. 

\begin{figure}[H]
	\centering
	\includegraphics[width=7.2cm]{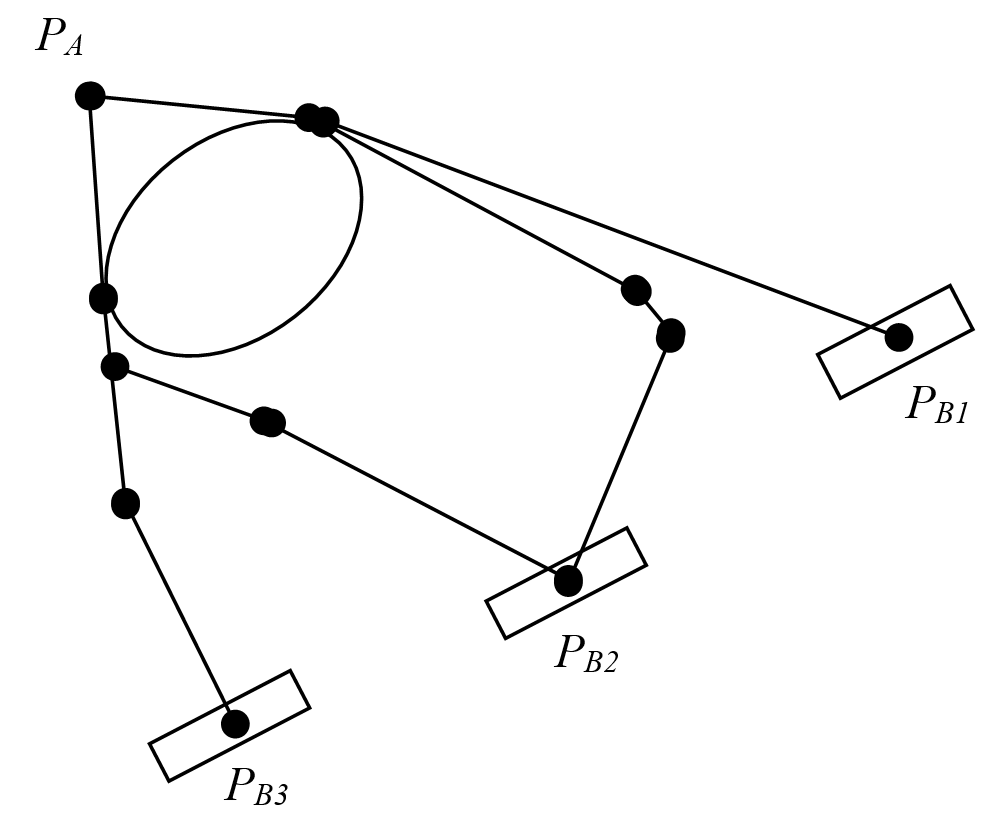}
	\caption{Illustration of the graph produced by the proposed algorithm. The obstacle shown is the local obstacle from $P_A$. Obstacles seen from other vertices are not depicted}\label{fig:ResultingGraphFollowingOGS}
	\par
\end{figure}
Note that in this subsection we have discussed the implementation of the algorithm starting from the engine cutoff point. Inside the aircraft, the algorithm will be invoked repeatedly every 300 ft of altitude loss and thus will consider changes in the environment.

\section{Accounting for the Effect of Turns}
\label{sec:EffectOfTruns}
In this section we augment our approach to account for the effect of turns. To this end, we simply add an estimate of the altitude loss associated with each change of heading, without modifying the geometry of the search graph construction. While the effect of a few turns may not be cardinal in a glide of several miles of more, this addition does provide better estimates of the total altitude loss of the selected trajectory.

We assume that, in practice, the pilot will perform a simple turn maneuver rather than an exact ALO maneuver which might be hard to carry out. In particular, we restrict our attention to fixed bank-angle maneuvers. In Appendix \ref{subsec:AltitudeLossDuetoTurns} we derive an analytical expression for the altitude-loss associated with fixed bank-angle maneuvers. This results in a circular arc in wind coordinates. We also show that the ALO turn in still air is performed at the stall velocity limit, $V_{stall}(n(\phi))$, which results in the minimal turn radius. The resulting minimal altitude loss, $\Delta Z ^*$, is given by
\begin{equation}\label{eq:AltitudeLossEquationBody}
\Delta Z^* =
\frac{2K_{SR}}{g}\left( \frac{V_{stall}(1)^4+V_0^4}{\sin (2 \phi)} \right) |\Delta \psi | + \frac{1}{2g}\left( V_{g,1}^2 - V_{g,0}^2  \right)	
\end{equation}
where $\Delta \psi$ is the heading change in the Airmass coordinates, $\phi$ the selected bank angle, and $V_{g,0}$ and $V_{g,1}$ are the ground velocities before and after the turn maneuver (to be taken as the ALO glide velocities at the respective straight-glide segments). The constants $V_{stall}$, $V_0$ and $K_{SR}$ are specified in Appendix \ref{subsec:AerodynamicModel}. We note that the second term in (\ref{eq:AltitudeLossEquationBody}) follows from the change in energy due to velocity change before and after the turn.

It may be seen from the above equation that minimal altitude loss is obtained for bank angle $\phi=\pi/4$. This value is to be selected, unless this angle is constrained to a smaller value.

The above estimate can be used in two ways. First, for a given glide trajectory composed of straight segments, the additional altitude loss can be computed consecutively for each heading change, and as needed it may be checked whether the modified altitude along the trajectory satisfies that ground clearance requirements. Alternatively, altitude corrections due to heading change may be incorporated in the search procedure in a straightforward manner by applying them to each newly-explored node, by so that they are taken into account during trajectory selection and optimization.

In Appendix \ref{subsec:NumericalExperiments} we compare the altitude-loss of the proposed turn segment to the optimal one computed by the GPOPS optimization package \cite{ref:GPOPS}. 
A related work concerning altitude loss during turns can be found in \cite{Paul18}.

\section{Reachability Scenarios}\label{sec:DemoScenarios}
Employing our generalized approach, we compute optimal maximum-range trajectories from the engine cutoff location A as a function of wind and the initial flight velocity and heading. In the following scenarios, we assume that the aircraft, Cessna 172S, weight is about 907 kg (no fuel -- faulty fuel gauge) \cite{ref:CessnaManual}). 

\emph{First scenario:}\label{subsec:FirstDemoScenario}
In the first demonstration, the aircraft has experienced engine cutoff at the horizontal position $P_A$ in Figure \ref{fig:LandingStripAttainability}, altitude 2500 m. The onboard emergency trajectory planning algorithm analyzes online the attainable landing site candidates. The algorithm employs the aircraft wind estimation
capability, yielding a 20 m/sec wind, heading North. The ALO manifold -- the surface surrounding $P_A$, depicted in Figure \ref{fig:LandingStripAttainability}, is then displayed on the LCD screen of the aircraft "Glass Cockpit" (GC), which presents the relevant obstacles. Also, the two landing sites are situated behind terrain obstacles inside the curve at the intersection of the ALO manifold and the terrain. The algorithm calculates the optimal trajectories to both landing sites. It turns out that the paved runway $B[2]$ at position $P_B[2]$ is unattainable as it is \hlhl{outside} the ALO manifold; thus, the pilot has no choice but to aim at landing site $B[1]$ which is a barren field. The algorithm yields the optimal trajectory from position $P_A$ to landing site $B[1]$, the solid red curve in Fig.\ \ref{fig:LandingStripAttainability}, and generates flight instructions aiding the pilot to follow the optimal trajectory to this second-best landing site.
\begin{figure}[H]
	\centering
	\includegraphics[width=8.2cm]{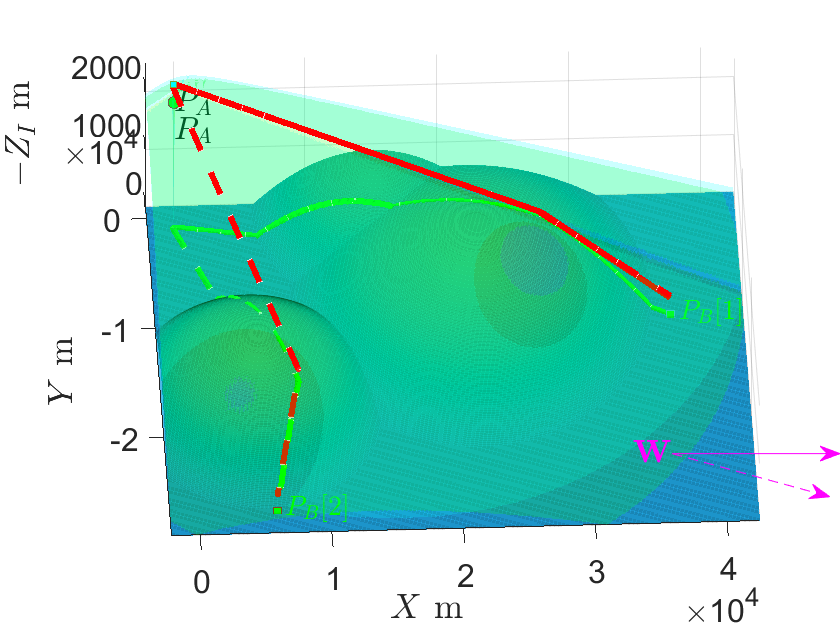}
	\caption{Reachability analysis for $\angle \mathbf{W}\in \{(0^\circ,-20^\circ)\} \wedge |\mathbf{W}| = 20 m/sec$ and obstacles 3D view}\label{fig:LandingStripAttainability}
	\par
\end{figure}

In case the wind estimator yields wind heading of -20$^\circ$, the better landing site, the paved runaway $B[2]$ can be reached. Again, the algorithm yields flight instructions, guiding the pilot to the optimal landing site.

\emph{Second scenario:}\label{subsec:SecondDemoScenario}
The aircraft experiences an engine malfunction at position $P_A$, as illustrated in Figure \ref{fig:SecondScenarioTrajectories}. The onboard emergency trajectory planning algorithm displays the altitude-loss manifold, following online wind estimation of a 15 m/sec wind heading West. The resulting manifold is then displayed on the GC screen -- the solid surface in Figure \ref{fig:SecondScenarioTrajectories}. In this scenario, the paved runway $B_2$ in position $P_B[2]$ is attainable.
The algorithm calculates the optimal trajectory from position $P_A$ to landing site $B_2$ and generates flight instructions to aid the pilot. \hlhl{After descent of about 10 km,} at position $P_A'$, the wind estimator detects that the wind magnitude has diminished to about 2 m/sec. The algorithm detects that the mountain ahead is now an obstacle that must be avoided. The onboard algorithm calculates the new optimal trajectory to $B_2$ that circumvents the obstacle, yielding the optimal velocity and channeling flight instructions directly to the GC, guiding the pilot safely to this landing site.
Our online algorithm re-directs the aircraft to circumvent the obstacle on the solid red trajectory towards $B[2]$.
\begin{figure}[H]
	\centering
	\includegraphics[width=8.2cm]{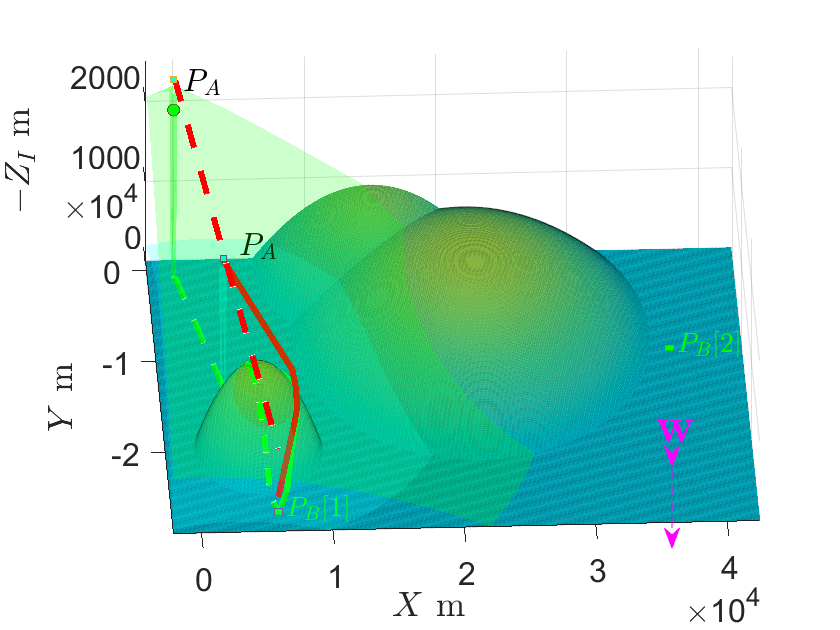}
	\caption{Reachability analysis $|\mathbf{W}|=15 m/sec\to 2 m/sec$ and obstacles 3D view}\label{fig:SecondScenarioTrajectories}
	\par
\end{figure}

\section{Conclusion}\label{sec:Conclusion}
In this work, we have derived the theoretical basis \hlhl{and the} algorithmic framework for calculating altitude-loss-optimal descent paths towards a \hlhl{candidate} landing location in case of engine power loss. The algorithm takes into account the effects of intense in-plane and crosswinds while avoiding ground-induced obstacles. 
\hlhl{First, our algorithm iteratively obtains intersections of altitude-loss-optimal manifolds, drawn from instantaneous aircraft locations with terrain elevation mapping; from these on-line-generated contours we construct sparse grids for OGS algorithmics.} The algorithm relies on a novel iterative visibility graph framework, which effectively turns the 3D problem into a sequence of 2D ones. \hlhl{I.e., we show that our algorithm needs to consider just two points on each intersection contour to optimally bypass any obstacle of whatever shape.} This serves to reduce the computational load, to allow for real-time calculation in case of emergency. \hlhl{We have proven that our algorithm is globally optimal in terms of altitude loss, subject to combined effects of winds and terrain-induced obstacles. We further include the effect of turns to assure safe near-optimal glide trajectories.

We apply our algorithm in realistic scenarios, using a Cessna 172 model, and demonstrate both the altitude-loss-optimal trajectory calculation and airstrip reachability determination. Furthermore, an initial validation flight text was conducted, as described in Appendix \ref{subsec:Flight}.

Note that our modeling assumes altitude-independent wind velocity and air density.
Relaxing the constant wind assumption invokes two challenges: (a) the ALO trajectory in free space is no longer a fixed-heading trajectory, 
(b) the local obstacle map may change in a way that it is not possible to reduce the 3D problem into local 2D problems. 
One may resort to solving numerically, using a piecewise-linear approximation of wind as function of altitude.  
Further research may relax the constant air density assumption as well.}



\section*{Acknowledgement}
We thank Mr.\ Yuval Dvir, a fully-certified flight test pilot, \hlhl{for piloting our validation flight testing,} and for his valuable insights. 
This research was supported by the Israel MoD Grant number 4441016309.

\appendices

\renewcommand{\theequation}{\thesection.\arabic{equation}}


\section{The Aerodynamic Model}
\label{subsec:AerodynamicModel}
\setcounter{equation}{0}

We start with modeling, the engine cutoff problem, as in \cite{ref:Segal}:
\begin{align}
\dot{X} &= V \cos(\gamma)\cos(\psi) + W_{X}\label{eq:X_dot}\\
\dot{Y} &= V \cos(\gamma)\sin(\psi) + W_{Y}\label{eq:Y_dot}\\
\dot{Z} &= - V \sin(\gamma)\label{eq:Z_dot}\\
\dot{V}&=-g\cdot \left( \frac{D}{mg}+\sin(\gamma) \right)\label{eq:VelocityDotEq}\\
\dot{\gamma}&=\frac{g}{V} \cdot \left( \frac{L\cos(\phi)}{mg}-\cos(\gamma) \right)\label{eq:GammmaDotEq}\\
\dot{\psi}&=\frac{L\sin(\phi)}{mV\cos(\gamma)}\label{eq:xiDotEq}
\end{align}
In Equations (\ref{eq:X_dot})-(\ref{eq:xiDotEq}) the aircraft is modeled as a point mass in a Ground frame of reference. 
The variables X,Y,Z are the North, East, Down location components of the point mass in the Ground frame. The variables $\psi$, $\gamma$, and $V$ are the heading, vertical angle and magnitude of the aircraft velocity vector relative to the air-mass.

The lift and drag forces are specified by the standard expressions:
\begin{equation*} \label{parabolic}
L = q S C_L\;,D =q S \left( C_{D0}+KC_L^2 \right)
\end{equation*}
where $q =\frac{1}{2} \rho V^2$. \hlhl{Recall that $C_L$ depends on the Angle of Attack.}

\hlhl{As in} \cite{ref:Segal}, we employ a reduced-order model where the fast variables are the true air velocity, $V$, the FPA and the bank-angle, $\phi$, while the slow variables are $X$, $Y$, $Z$ and $\psi$.
\hlhl{The control variables in this model are the pair $(V(t),\phi(t))$.}

Substituting $\dot{\gamma} \cong 0$ into Equation (\ref{eq:GammmaDotEq}) yields the load factor $n=\frac{L}{mg}$ as a function of $\phi$ and $\gamma$:\footnote{Equations (\ref{eq:LoadFactorVsPhi}) (\ref{eq:VstallEq}), (\ref{eq:BasicSinkRateEquation})), (\ref{eq:OptimalVelocityinStillAir}) are analogous to \cite[Equations (9.67), (8.18), (9.73)]{ref:Vinh} and \cite[Equation (3.6)]{ref:Irving}}
\begin{equation}\label{eq:LoadFactorVsPhi}
n(\phi,\gamma)=\frac{\cos(\gamma)}{\cos(\phi)}
\end{equation}
Thus, $n \cong n(\phi) = \frac{1}{\cos(\phi)}$ (e.g., \cite[Equation (4.21)]{ref:Irving}).

The stall limit of the aircraft is obtained by employing the load factor definition $n=\frac{L}{mg}$ and the lift force equation at $C_L=C_{Lmax}$:
\begin{equation}\label{eq:VstallEq}
V_{stall}(n(\phi))=\sqrt{\frac{2mg}{\rho S C_{Lmax}}n(\phi)}
\end{equation}
Combined with an upper limit $V_{max}$ on the flight velocity, we have $V_{stall}(n(\phi)) \leq V \leq V_{max}$.

The cost function in our problem is the altitude loss, which is the integral of the sink-rate. Starting from Equation (\ref{eq:Z_dot}), the sink rate can be expressed as a function of the control variables, $(V,\phi)$, using the parabolic drag approximation (\ref{parabolic}) and  employing (\ref{eq:VelocityDotEq}) subject to $\dot{V} \cong 0$, yielding
\begin{equation}\label{eq:BasicSinkRateEquation}
\dot{Z}= \frac{\rho SC_{D0}}{2mg}V^3+\frac{2Kmg}{\rho S}\frac{n^2}{V}
\end{equation}
It is convenient to express the sink rate in terms of $V_0$, the optimal max-range glide velocity in still air. 
The optimal max-range-optimal dynamic pressure in still air, subject to small FPA approximation, is given by $q_0=\frac{mg}{S}\sqrt{\frac{K}{C_{D0}}}$.
Therefore, $V_0$, the optimal glide velocity in still air, is:
\begin{equation}\label{eq:OptimalVelocityinStillAir}
V_0=\sqrt{\frac{2mg}{\rho S}\sqrt{\frac{K}{C_{D0}}}}
\end{equation}
The sink rate  can now be expressed in terms of $V_0$ and $\phi$ by substituting Equation (\ref{eq:OptimalVelocityinStillAir}) and $n \cong n(\phi)$ into the sink rate Equation (\ref{eq:BasicSinkRateEquation}). We denote the sink rate function as $f_0(V,\phi)$:
\begin{equation}
\dot{Z} = f_0(V,\phi)=K_{SR}\left( \frac{V^4+n(\phi)^2V_0^4}{V} \right)\label{eq:ImprovedSinkrateEquation}
\end{equation}
where $K_{SR}= \frac{\rho S C_{D0}}{2mg}$ and $n(\phi) = \frac{1}{\cos(\phi)}$. In straight glide segments $f_0(V,0)$ defines the running cost of our problem. Minimizing the cumulative running cost constitutes our optimization objective.

\section{Altitude Loss Due to Turns}
\label{subsec:AltitudeLossDuetoTurns}
\setcounter{equation}{0}

Let us derive the expression for altitude loss as a result of heading change. 
We work in the Airmass Frame. Employing Equations (\ref{eq:xiDotEq}), (\ref{eq:LoadFactorVsPhi}) and $n=\frac{L}{mg}$, yields the heading turn rate:
\begin{equation}\label{eq:QuasiXiDotTan}
\dot{\psi}=\frac{g}{V} \tan (\phi)
\end{equation}
Let us obtain the ALO maneuver. From the chain rule of derivation, 
\begin{equation*}
\dot{Z}=\frac{\partial Z}{\partial \psi}\dot{\psi} \implies \frac{\partial Z}{\partial \psi}= \frac{\dot{Z}}{\dot{\psi}}
\end{equation*}
Now, employing $n=\frac{1}{\cos(\phi)}$ and Eqs.\ (\ref{eq:ImprovedSinkrateEquation}, \ref{eq:QuasiXiDotTan}), for $\dot{Z}$ and $\dot{\psi}$:
\begin{equation}\label{eq:AltitudeLossForTurns}
\frac{\partial Z}{\partial \psi}=\frac{2K_{SR}}{g}\left( \frac{V^4 cos(\phi)^2 +V_0^4}{\sin (2 \phi)} \right)
\end{equation}
As $\frac{\partial Z}{\partial \psi}$ is monotonous in $V$, the ALO turn velocity is the lower bound, $V_{stall}(n(\phi))$. Let us substitute Eq.\ (\ref{eq:VstallEq}) into Eq.\ (\ref{eq:AltitudeLossForTurns}) to obtain the minimal altitude-loss, $\frac{\partial Z}{\partial \psi}^*$, at bank angle $\phi$:
\begin{equation}\label{eq:AltitudeLossForTurnsStallLimit}
\frac{\partial Z}{\partial \psi}^*=\frac{2K_{SR}}{g}\left( \frac{V_{stall}(1)^4+V_0^4}{\sin (2 \phi)} \right)
\end{equation}
Further, the latter expression is clearly minimal for $\phi=\frac{\pi}{4}$.
Therefore, the ALO maneuver is to turn at constant velocity, $V_{stall}(n(\frac{\pi}{4}))$, and a bank angle, $\phi=\frac{\pi}{4}$.
This result corresponds the 'controls histories' in \cite[Figures 9(c),11(c),12(c)]{ref:WolekGlider}, which were acquired by numerical integration of extremals. The velocity control change to the stall limit while increasing the load factor, $n$, when maneuvering; however, in these examples, the load factor increases to more than $n(\frac{\pi}{4})$ in order to meet the boundary constraints. 

With constant $\phi$ and $V$, the aircraft performs the turn with a constant radius in the Airmass Frame, namely  
\begin{equation}\label{eq:TurnRadius}
R = \frac{V^2}{g \tan(\phi)}
\end{equation}
\begin{figure}[H]
	\centering
	\includegraphics[width=6.2cm]{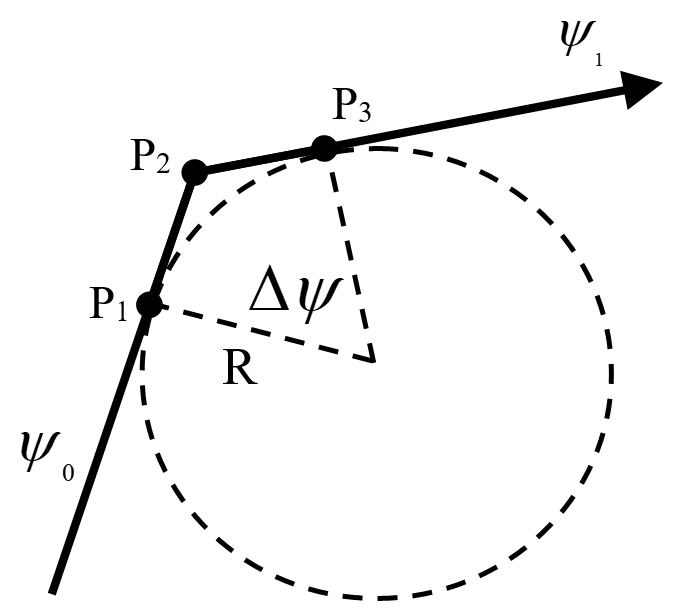}
	\caption{The spatial geometry of turns in the Airmass Frame}\label{fig:TurnGeometry}
	\par
\end{figure}
For a heading change from $\psi_0$ to $\psi_1$ (see Figure \ref{fig:TurnGeometry}), the altitude loss during the turn itself, as follows from (\ref{eq:AltitudeLossForTurnsStallLimit}), is given by
\begin{equation}
(\Delta Z)_a =
\frac{2K_{SR}}{g}\left( \frac{V_{stall}(1)^4+V_0^4}{\sin (2 \phi)} \right) |\psi_1 -\psi_0|
\end{equation}

In addition to the altitude-loss during the turn itself, we need to consider the effect of the velocity change before and after the maneuver. Note that under our assumptions, the flight velocity is a fast time-scale variable, so that the aircraft can change its velocity instantaneously when it enters and exits the the turn maneuver.
This variation in kinetic energy therefore directly translates to variation in potential energy, leading to the following estimate for  the total altitude loss:
\begin{equation}\label{eq:AltitudeLossEquation}
\Delta Z^* = (\Delta Z)_a 
+ \frac{1}{2g}\left( V_{g,1}^2 - V_{g,2}^2 \right)
\end{equation}
where $V_{g,0}$ and $V_{g,1}$, respectively, are the aircraft ground velocities before and after the turn maneuver.



\section{Performance Analyses}
\label{subsec:NumericalExperiments}
\setcounter{equation}{0}


\hlhl{We present here sample sensitivity analyses for the formulations of the previous subsection for altitude-loss due to the turn maneuver. We employ the aerodynamic model of a Cessna 172.}
We fitted a quadratic drag approximation to the model \hlhl{in} \cite{ref:CessnaDataSheet}. The resulting aircraft model parameters are: $C_{D0}$ = 0.0329, $K$ = 0.0599 with mass $m$ = 907 kg, S = 15.9793 m$^2$, $V_{stall}(1) $ = 27.27 m/sec. We further used a standard atmosphere model with 15\si{\degree}C at sea level.

\emph{Turn Performance Analysis:}
Figure \ref{fig15} depicts the altitude loss rate according to Eq.\  (\ref{eq:AltitudeLossForTurns}) as a function of the turn speed and bank angle. Figure \ref{fig16} shows the turn radius per Eq.\ (\ref{eq:TurnRadius}). We observe that indeed the optimal solution is to fly at stall limit and a bank angle of 45 degrees.
\begin{figure}[H] 
	\centering
	\includegraphics[width=8.2cm]{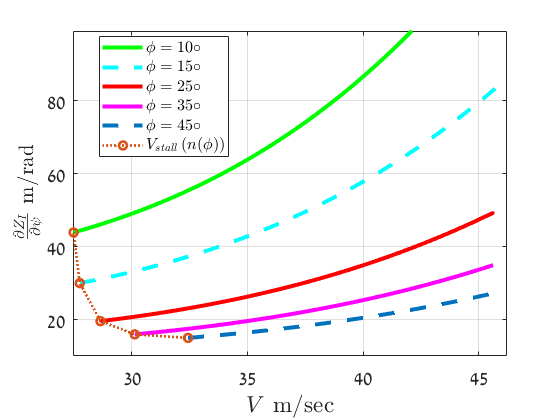}
	\caption{Altitude loss vs.\ the turn velocity in the Airmass Frame}
	\par
	\label{fig15}
\end{figure}
\begin{figure}[H] 
	\centering
	\includegraphics[width=8.2cm]{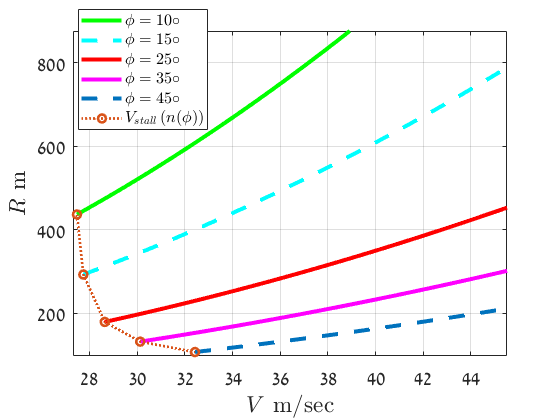}
	\caption{Turn radius vs.\ the turn velocity in the Airmass Frame}
	\par
	\label{fig16}
\end{figure}
 
\emph{Comparison to a GPOPS optimal solution:}
We proceed to demonstrate our analytic solution for ALO maneuver vs. the numerical solution of GPOPS \cite{ref:GPOPS}. 
The GPOPS solver was given the full state Eqs. (\ref{eq:X_dot}-\ref{eq:xiDotEq}) without our assumptions of small glide angles.
Ours is an analytical approximation, assuming a step change in velocity (and consequently in energy), whereas the GPOPS search is quasi-continuous.
\begin{figure}[H]
	\centering
	\includegraphics[width=8.2cm]{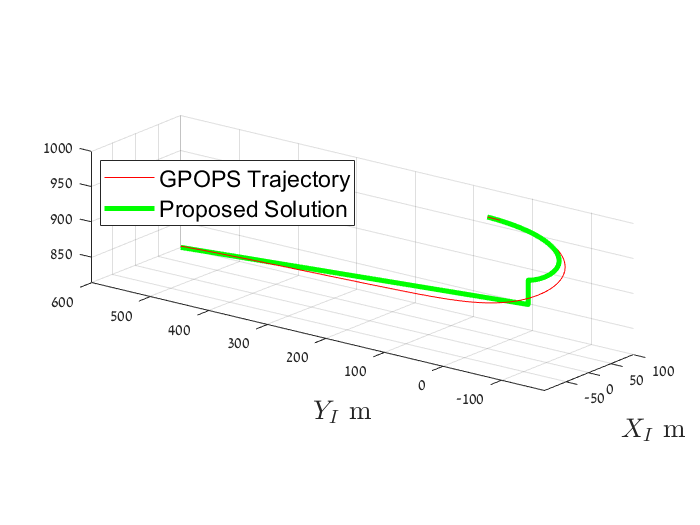}
	\caption{GPOPS vs analytic solution 3D view trajectory with $W_{Y}= -20$ and $W_{X}= 0$}\label{fig:GpopsWithWindXYZ}
	\par
\end{figure}
Fig.\ \ref{fig:GpopsWithWindXYZ} shows a good match between our solution and the optimal-turn trajectory obtained via GPOPS. 
\hlhl{Note that in our simulation we have assumed a step change in velocity vector direction (in Section} \ref{sec:EffectOfTruns} \hlhl{we discuss optimal concurrent descent and turn maneuvering). GPOPS numerical scheme does not assume such a step-change in the V direction.}

\section{Flight Experiment}
\label{subsec:Flight}
\setcounter{equation}{0}

Towards a validation of the proposed model and algorithms, a flight experiment plan was devised.
This Appendix describes briefly the initial experiment that was conducted.

We chose to flight-test our optimal algorithm on a Cessna 172 -- to demonstrate both the optimal airstrip choice and the trajectory generation, and the following of this trajectory by pilot-in-distress.

A team of undergraduate students pursued this objective, supervised by Daniel Segal, and advised by Dr.\ Aharon Bar-Gill. Technical assistance was provided by the Technion CRML laboratory (Control, Robotics and Machine Learning) and its engineering team, headed by Mr.\ Koby Kohai.
The following text was provided by A.\ Bar-Gill.

\begin{figure}[H] 
	\centering
	\includegraphics[width=11.2cm]{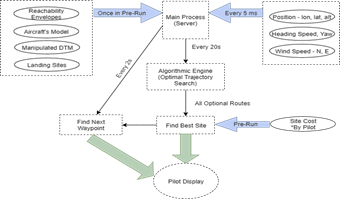}
	\caption{System schematics –- on-line algorithmic computation}
	\par
	\label{fig20:Schematics}
\end{figure}

The team has developed a dedicated simulation for testing the algorithm implementation. It involves flight modeling, generation of optimal trajectory towards the preferable airstrip and cues on a screen - for the pilot to track this trajectory. The offline simulation was then adapted, its software was embedded into the airborne PC and run in the lab as a hybrid simulation –- for debugging the algorithm implementation and firmware interfaces.

\begin{figure}[H] 
	\centering
	\includegraphics[width=11.2cm]{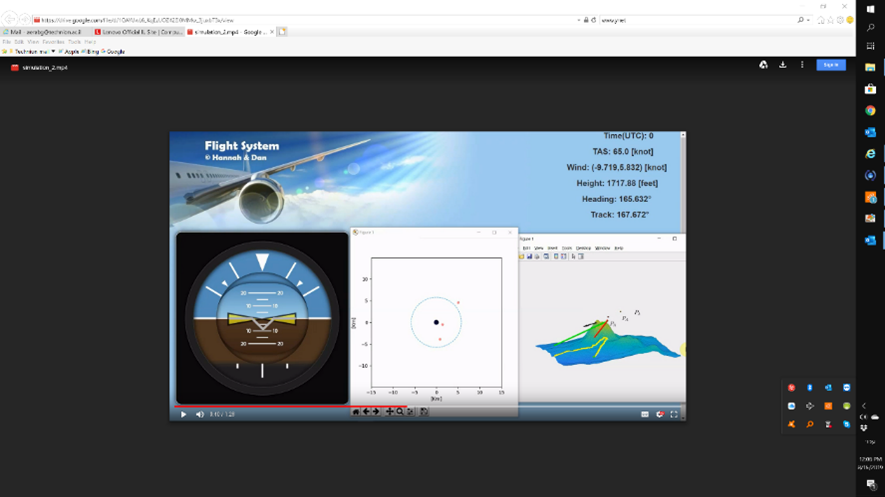}
	\caption{Pilot Dedicated Display}
	\par
	\label{fig21:Display}
\end{figure}

The simulated scenario is shown on screens of both the system portable PC and the pilot’s dedicated display (used in course of our flight demo):
\begin{itemize}
\item[(a)]
The right side of the pilot display depicts the dynamically-generated WP’s (Way Points). Following these WP’s, the pilot circumvents terrain elevations – on his way to the goal airstrip. The algorithm repeatedly checks the validity of the originally-chosen optimal trajectory (green), with black arrow pointing along – vis-à-vis th e red ones, which represent trajectories towards higher-weighing candidate airstrips
In the simulated scenario, the aircraft velocity vector’s pointing tracks the green trajectory by following the dynamic Way Point, running along the trajectory and periodically comparing with the backup trajectory (red).
\item[(b)]
The central,  radar-like, display screen– the black circle at its center symbolizes the aircraft, and the red points – the candidate landing strips. The blue “disk” represents the instantaneous reachability envelope.
\item[(c)] 
The pilot’s guidance cues –- heading and roll to point at along-trajectory-running WP and pitch angle / Descent Rate / V\textunderscore air.
The left side of the pilot display, the artificial horizon – “V”-symbol with “wings” stands for instantaneous pointing of the velocity vector of the aircraft. In order to track the optimal trajectory, computed by the researchers’ algorithm, the pilot must align this “V with wings” symbol with the symbol, which comprises two yellow triangles.
\end{itemize}

\noindent\emph{The overall system:}
\begin{itemize}
\item
VectorNav unit, onboard computer, pilot display, Client-Server based on UDP connection, next waypoint in shared file.
\end{itemize}

\noindent \emph{Real-Time Inputs:}
\begin{itemize}
\item
Position (Lon, Lat, Alt), Heading (Speed, Yaw), Wind Speed (North, East).
\end{itemize}

\noindent \emph{Pre-Processing:}
\begin{itemize}
\item DTM -– high resolution elevation map array, sliced into 3 parts.
\item Aircraft model – CSV file, pre-calculation of parameters for quick access.
\item Reachability Envelopes -– calculation yielding optimal altitude loss trajectories.
\item Landing sites – CSV file, containing optional landing sites.
\end{itemize}

\noindent\emph{Algorithm:}
\begin{itemize}
\item Reachable Landing Sites – landing sites within the chosen area. 
\item Best Landing Site – the site with the highest attribute. 
\item W.P.-to-W.P. Guidance, following the Next Waypoint calculation.
\item Pilot route visualization - using an artificial horizon.
\end{itemize}

At the airfield, the students’ team installed the airborne system, and inputted into the system PC the locations of the candidate airstrips (East of Mount Tabor) and their respective weightings. 

\emph{The Flight Test Team (professional flight-test crew, volunteering their expertise):}
Test pilot, Yuval Dvir tracked the optimal trajectory cue on a dedicated display, 
Safety pilot, Mike Dvir maintained the glide-optimal ﬂight velocity (value above Idle), and kept checking the sanity of the algorithm implementation.

\begin{figure}[H] 
	\centering
	\includegraphics[width=11.2cm]{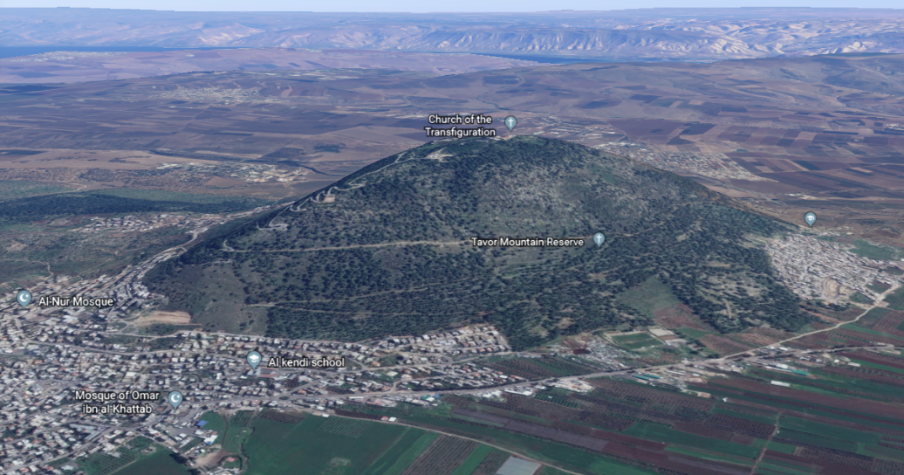}
	\caption{Mount Tabor}
	\par
	\label{fig22:Tabor}
\end{figure}

\emph{Flight test scenario -- July 30, 2019:} 
Assuming engine failure at altitude of 2200 [ft] and heading 100$^{\circ}$  West of Mount Tabor, and potential landing strips East of it, the algorithm has chosen the best landing strip available East of Mount Tabor and computed the optimal trajectory for the test pilot to track. And indeed, the pilot has circumvented Mount Tabor appropriately, from North, towards the virtual landing strip North-East of Mount Tabor, preferred over another virtual site South-East of Tabor. 
As a result, the concept was in-flight validated, implementing our real-time algorithm, and tracking the globally-optimal trajectory by the pilot.


\begin{thebibliography}{}

\bibitem{ref:Segal} 
D. Segal, A. Bar-Gill and N. Shimkin, 
``Max-Range Glide in Engine Cutoff Emergencies Under Severe Wind", 
\textit{Journal of Guidance, Control, and Dynamics,} Vol. 42(8), pp. 1822-1835, August 2019. 

\bibitem{ref:Adler} 
A. Adler, A. Bar-Gill and N. Shimkin,  
``Optimal Flight Paths for Engine-out Emergency Landing using Flight Primitives", 
\textit{Proc. 24th Chinese Control and Decision Conference (CCDC)}, Taiyuan, China, pp.\ 2908-2915,  May 2012.

\bibitem{ref:Fixed-WingUAV} 
X. Fang, N. Wan, H. Jafarnejadsani, D. Sun, F. Holzapfel and N. Hovakimyan, 
``Emergency Landing Trajectory Optimization for Fixed-Wing UAV under Engine Failure", 
\textit{Proc. AIAA Scitech Forum,} San Diego, CA, pp. 7–11, January 2019. 

\bibitem{ref:HeuristicGeneticAlgorithmApproaches} 
J. da Silva Arantes, M. da Silva Arantes, C. F .M. Toledo, O. J. Trindade and C. B. Williams, 
``Heuristic and Genetic Algorithm Approaches for UAV Path Planning under Critical Situation", \textit{International Journal on Artificial Intelligence Tools,} Vol. 26, No. 1, 2017.

\bibitem{ref:EvaluatingHardwarePlatforms}
J. da Silva Arantes, M. da Silva Arantes, A. B. Missaglia, E. do Valle Simoes and C. F. M. Toledo,  
``Evaluating Hardware Platforms and Path Re-planning Strategies for the UAV Emergency Landing Problem",
\textit{Proc. IEEE 29th International Conference on Tools with Artificial Intelligence,}
Boston, MA, pp. 937–944, November 2017.

\bibitem{ref:RRT_Star_AR} 
S. Choudhury, S. Scherer and S.\\ Singh, 
``RRT*-AR: Sampling-Based Alternate Routes Planning with Applications to Autonomous Emergency Landing of a Helicopter", 
\textit{IEEE International Conference on Robotics and Automation}, May 2013.

\bibitem{ref:EmergencyLandingGuidance} 
J. Slama, ``Emergency Landing Guidance for an Aerial Vehicle with Motor Malfunction", Thesis, Faculty of Electrical Engineering, Czech Technical University, Prague, 2018.




\bibitem{ref:ReachabilityBasedLanding} 
A. K. Akametalu, C. J. Tomlin and M. Chen, 
``Reachability-Based Forced Landing System", 
\textit{Journal of Guidance, Control and Dynamics,} Vol. 41, No. 12, Dec. 2018.

\bibitem{ref:MeuleauEmergencyLandingPlanner2009a} 
N. Meuleau, C. Plaunt, D. E. Smith and T. Smith, 
``A comparison of risk sensitive path planning methods for aircraft emergency landing", 
\textit{Proc. ICAPS'09 Workshop on Bridging The Gap Between Task And Motion Planning,} pp. 71–80, 2009.

\bibitem{ref:MeuleauEmergencyLandingPlanner2009b} 
N. Meuleau, C. Plaunt, D. E. Smith and T. Smith, 
``An Emergency Landing Planner for Damaged Aircraft",
\textit{Proc. 21st Conference on Innovative Applications of Artificial Intelligence (IAAI-09),} pp. 114-121, AAAI Press, 2009. 

\bibitem{ref:MeuleauEmergencyLandingPlannerExpriment} 
N. Meuleau, C. Neukom, C. Plaunt, D. E. Smith and T. Smith, 
``The emergency Landing Planner Experiment",
\textit{Proc. ICAPS Scheduling and Planning Application Workshop (SPARK'11),} pp.\ 60-67, 2011.  

\bibitem{ref:AircraftLossOfControl2017} 
C. M. Belcastro, J. V. Foster, G. H. Shah, I. M. Gregory, D. E. Cox, A. A. Crider, L. Groff, R. L. Newman and D. H. Klyde, 
``Aircraft Loss of Control Problem Analysis and Research Toward a Holistic Solution", 
\textit{Journal of Guidance, Control and Dynamics}, Vol.\ 40, No.\ 4, pp.\ 733-775, 2017. 

\bibitem{ref:VoronoiDiagrams} B. Grüter, D. Seiferth, M. Bittner and F. Holzapfel, 
``Emergency Flight Planning using Voronoi Diagrams", 
\textit{AIAA Scitech Forum}, article 1056, San Diego, CA, January 2019.

\bibitem{ref:TotalLossofThrust} E. M. Atkins, I. A. Portillo and M. J. Strube, 
``Emergency Flight Planning Applied to Total Loss of Thrust", 
\textit{Journal of Aircraft,} Vol. 43, No. 4, July-August 2006.

\bibitem{ref:MobileRobotBoundedAccelerations} 
J. Ben-Asher, M. Wetzler, E. D. Rimon and J. Diepolder, 
``Optimal trajectories for a mobile robot with bounded accelerations in the presence of a wall or a bounded obstacle", 
\textit{Proc. 27th Mediterranean Conference on Control and Automation (MED)}, pp.\ 481–488, 2019.

\bibitem{Scholer11}
F. Scholer, A. la Cour-Harbo and M. Bisgaard, 
``Configuration space and visibility graph generation from geometric workspaces for UAVs", 
\textit{Proc. AIAA Guidance, Navigation, and Control Conference,} Portland, Oregon, paper 6416, August 2011. 

\bibitem{Maini16}
P. Maini, and P. Sujit, 
``Path planning for a UAV with kinematic constraints in the presence of polygonal obstacles", 
\textit{Proc. Int. Conference on Unmanned Aircraft Systems (ICUAS),} pp. 62–67, 2016.

\bibitem{Frontera17}
G. Frontera, D. J. Martín, J. A. Besada and D. Gu, 
``Approximate 3D Euclidean Shortest Paths for Unmanned Aircraft in Urban Environments", 
\textit{Journal of Intelligent and Robotic Systems,} Vol. 85, pp. 353–368, 2017.

\bibitem{Ahmad17}
Z. Ahmad, F. Ullah, C. Tran and S. Lee, 
``Efficient Energy Flight Path Planning Algorithm Using 3-D Visibility Roadmap for Small Unmanned Aerial Vehicle", \textit{International Journal of Aerospace Engineering,} Vol. 2017, Article ID2849745, 2017. 

\bibitem{Majeed18}
A Majeed, and S. Lee, 
``A fast global flight path planning algorithm based on space circumscription and sparse visibility graph for unmanned aerial vehicle", \textit{Electronics,} Vol. 7(12), Article 375, 2018.

\bibitem{Damato18}
E. D'Amato, I. Notaro, and M. Mattei, 
``Optimal Flight Paths over Essential Visibility Graphs", 
textit{Proc. International Conference on Unmanned Aircraft Systems (ICUAS),} Dallas, TX, pp. 708-714, 2018,  

\bibitem{Damato19}
E. D’Amato, I. Notaro, L. Blasi and M. Mattei, 
``Smooth Path planning for Fixed-Wing Aircraft in 3D Environment Using a Layered Essential Visibility Graph", 
\textit{Proc. International Conference on Unmanned Aircraft Systems (ICUAS),} Atlanta, GA, 2019, pp. 9-18. 

\bibitem{Huang19}
S. Huang and R. S. H. Teo, 
``Computationally Efficient Visibility Graph-Based Generation of 3D Shortest Collision-Free Path Among Polyhedral Obstacles for Unmanned Aerial Vehicles", 
\textit{Proc. International Conference on Unmanned Aircraft Systems (ICUAS),}  Atlanta, GA, 2019, pp. 1218-1223.

\bibitem{ref:HybridOptimalControl} 
R. Chai, A. Savvaris, A. Tsourdos, S. Chai and and Y. Xia, 
``Trajectory optimization of space maneuver vehicle using a hybrid optimal control solver", 
\textit{IEEE Trans. on Cybernetics}, Vol. 49, Issue 2, Feb. 2019.

\bibitem{ref:StochasticSpacecraftTrajectory} 
R. Chai, A. Savvaris, A. Tsourdos, S. Chai and Y. Xia, 
``Stochastic spacecraft trajectory optimization with the consideration of chance constraints", 
\textit{IEEE Trans. Control Syst. Technol.}, Vol. 28, Issue 4, July 2020.

\bibitem{ref:ProbabilisticConstraints} 
R. Chai, A. Savvaris, A. Tsourdos, S. Chai, Y. Xia, S. Wang, 
``Solving trajectory optimization problems in the presence of probabilistic constraints", 
\textit{IEEE Trans. Cybernetics}, Vol. 50, Issue 10, Oct. 2020.

\bibitem{ref:LaVallePllaningAlg} 
S. M. LaValle,
\textit{Motion Planning}, Cambridge University Press, Cambridge, 2006.



\bibitem{ref:USGS} 
"Earth Explorer", U.S. Geological Survey Website. https://earthexplorer.usgs.gov/


\bibitem{ref:Hart68}
P. E. Hart, N. J. Nilsson and B. Raphael, 
``A Formal Basis for the Heuristic Determination of Minimum Cost Paths", 
\textit{IEEE Transactions on Systems Science and Cybernetics,} Vol. 4, No. 2, pp. 100–107, 1968.

\bibitem{ref:Nilsson80}
N.\ J.\ Nilsson, \textit{Principles of Artificial Intelligence,}
Tioga Publishing Company, Palo Alto, CA, 1980.

\bibitem{ref:Liu94}
Y. Liu and S. Arimoto, 
``Computation of the Tangent Graph of Polygonal Obstacles by Moving-Line Processing",
\textit{IEEE Transactions on Robotics and Automation,} Vol. 10, No. 6, pp. 823-830, 1994.



\bibitem{ref:CessnaManual} 
\textit{Cesna 172S Skyhawk Information Manual}, Cessna Aircraft Company, Revision 5, Wichita, Kansas, July 2004.


\bibitem{ref:WolekGlider} A. Wolek, E. M. Cliff and C. A. Woolsey, 
``Energy-Optimal Paths for a Glider with Speed and Load Factor Controls", 
\textit{Journal of Guidance, Control and Dynamics}, Vol. 39, No. 2, p. 397-405, February 2016. 

\bibitem{Paul18}
S. Paul, F. Hole, A. Zytek and C. A. Varela,  
``Wind-Aware Trajectory Planning for Fixed-Wing Aircraft in Loss of Thrust Emergencies", 
\textit{Proc. IEEE/AIAA 37th Digital Avionics Systems Conference (DASC),} London, 2018.

\bibitem{ref:Vinh} N. X. Vinh, 
\textit{Flight Mechanics of High Performance Aircraft}, Cambridge University Press, Cambridge, England, 1993.

\bibitem{ref:Irving} F.\ Irving, 
\textit{The Paths of Soaring Flight}, Imperial College Press, London, 1999, pp.\ 17-32.

\bibitem{ref:CessnaDataSheet} J. Scott and M. Selig,  
``Cessna 172: Small Single Piston Engine General Aviation Airplane", Aircraft Dynamics Models for Use with FlightGear, Course Notes:, AAE319 Airplane Flight Dynamics, University of Illinois, Spring 2000, https://m-selig.ae.illinois.edu/apasim/Aircraft-uiuc/cessna172-aae319-v4/aircraft.dat.

\bibitem{ref:GPOPS} 
A. Rao, D. Benson, C. Darby, M. Patterson, C. Francolin, I. Sanders and G. Huntington, 
``Algorithm 902: GPOPS, a MATLAB Software for Solving Multiple-Phase Optimal Control Problems Using the Gauss Pseudospectral Method", 
\textit{ACM Transactions on Mathematical Software}, Vol. 37, No. 2, 2010.



\bibitem{ref:engine_failures_2016}
``Engine failures and malfunctions in light aeroplanes, 2009-2014", 
Australian Transport Safety Bureau, March 2016. 
https://www.atsb.gov.au/media/5769864/ar-2013-107-final-report.pdf

\bibitem{survey2010}
C. Goerzen, Z. Kong and B. Mettler, 
``A Survey of Motion Planning Algorithms from the Perspective of Autonomous UAV Guidance",  \textit{Journal of Intelligent Robot Systems,} Vol. 57, 2010.

\bibitem{survey2019}
Y. Yang, J. Pan and W. Wan, ``Survey of optimal motion planning",
\textit{IET Cyber-Systems and Robotics,} Vol. 1(1), 2019,  pp. 13-19. 9

\bibitem{survey2020}
L. Quan, L. Han, B. Zhou, S. Shen, and F. Gao, ``Survey of UAV motion planning", \textit{IET Cyber‐Systems and Robotics,} Vol. 2, May 2020, pp. 14-21.

\bibitem{Ayhan2019}
B. Ayhan, C. Kwan, B. Budavari, J. Larkin and D. Gribben, ``Preflight Contingency Planning Approach for Fixed Wing UAVs with Engine Failure in the Presence of Winds”, \textit{Sensors}, Vol. 19, January 2019.

\end{thebibliography}
\end{document}